\documentclass[11pt,reqno]{amsart}
\usepackage{amsmath,amsthm,amssymb,bm,bbm,dsfont,braket}
\usepackage[mathscr]{eucal}
\usepackage{amsaddr}
\usepackage{upgreek}
\usepackage[left=2cm,right=2cm,top=2cm,bottom=2cm]{geometry}
\usepackage{mathtools}\mathtoolsset{centercolon}
\mathtoolsset{showonlyrefs}
\DeclarePairedDelimiter{\ceil}{\lceil}{\rceil}
\usepackage{cite}

\usepackage[shortlabels]{enumitem}
\usepackage{setspace}
\usepackage{graphicx}
\usepackage{verbatim}
\usepackage{float}
\usepackage{placeins}
\usepackage{array}
\usepackage{booktabs}
\usepackage{threeparttable}
\usepackage[update,prepend]{epstopdf}
\usepackage{multirow}
\usepackage{amsfonts,amssymb,dsfont,xfrac,comment}
\usepackage[abs]{overpic}

\usepackage[usenames,dvipsnames]{color}
\usepackage[hidelinks]{hyperref}
\hypersetup{
	unicode=false,          
	pdftoolbar=true,        
	pdfmenubar=true,        
	pdffitwindow=false,     
	pdfstartview={FitH},    
	pdftitle={My title},    
	pdfauthor={Author},     
	pdfsubject={Subject},   
	pdfcreator={Creator},   
	pdfproducer={Producer}, 
	pdfkeywords={keyword1} {key2} {key3}, 
	pdfnewwindow=true,      
	colorlinks=true,        
	linkcolor=Red,          
	citecolor=ForestGreen,  
	filecolor=Magenta,      
	urlcolor=BlueViolet,    
}
\usepackage{doi}
\usepackage{url}
\usepackage{caption, subcaption}
\usepackage{enumitem}
\usepackage{shadethm}

\makeatletter
\ifx\@NODS\undefined%

\let\mathbb=\mathds
\else%
\fi
\makeatother

\DeclareMathOperator{\Tr}{Tr}

\DeclareMathOperator{\e}{\mathrm{e}}

\newcommand{\be}{{\mathbf e}}

\newcommand{\tr}{\operatorname{Tr}}

\newcommand{\ten}{\otimes}
\newcommand{\pl}{\hspace{.1cm}}

\def\0{{\mathbf{0}}}
\def\1{{\mathbf{1}}}
\def\2{{\mathbf{2}}}
\def\3{{\mathbf{3}}}
\def\4{{\mathbf{4}}}
\def\5{{\mathbf{5}}}
\def\6{{\mathbf{6}}}

\def\7{{\mathbf{7}}}
\def\8{{\mathbf{8}}}
\def\9{{\mathbf{9}}}

\def\be{\begin{equation}}
\def\ee{\end{equation}}
\def\bea{\begin{eqnarray}}
\def\eea{\end{eqnarray}}

\def\eps{\varepsilon}

\theoremstyle{plain}
\newtheorem{theo}{Theorem} 
\newtheorem{prop}[theo]{Proposition} 
\newtheorem{lemm}[theo]{Lemma} 
\newshadetheorem{theorem}{Theorem}

\theoremstyle{definition}

\theoremstyle{remark}
\newtheorem{remark}{Remark}[section]

\numberwithin{equation}{section}

\makeatletter
\newcommand{\opnorm}{\@ifstar\@opnorms\@opnorm}
\newcommand{\@opnorms}[1]{%
	$\left|\mkern-1.5mu\left|\mkern-1.5mu\left|
	#1
	\right|\mkern-1.5mu\right|\mkern-1.5mu\right|$
}
\newcommand{\@opnorm}[2][]{%
	\mathopen{#1|\mkern-1.5mu#1|\mkern-1.5mu#1|}
	#2
	\mathclose{#1|\mkern-1.5mu#1|\mkern-1.5mu#1|}
}
\makeatother

\begin{document}

\let\origmaketitle\maketitle
\def\maketitle{
	\begingroup
	\def\uppercasenonmath##1{} 
	\let\MakeUppercase\relax 
	\origmaketitle
	\endgroup
}

\title{\bfseries \Large{ 
Optimal Second-Order Rates for \\ Quantum Soft Covering and Privacy Amplification
		}}

\author{ \normalsize \textsc{Yu-Chen Shen$^{1}$, Li Gao$^{5}$, and Hao-Chung Cheng$^{1,2,3,4}$}}
\address{\small  	
$^{1}$Department of Electrical Engineering and Graduate Institute of Communication Engineering,\\ National Taiwan University, Taipei 106, Taiwan (R.O.C.)\\
$^{2}$Department of Mathematics, National Taiwan University\\
$^{3}$Center for Quantum Science and Engineering,  National Taiwan University\\
$^{4}$Hon Hai (Foxconn) Quantum Computing Center, New Taipei City 236, Taiwan\\
$^{5}$Department of Mathematics, University of Houston, Houston, TX 77204, USA\\
}

\email{\href{mailto:haochung.ch@gmail.com}{haochung.ch@gmail.com}}

\date{\today}
\begin{abstract}
We study quantum soft covering and privacy amplification against quantum side information. The former task aims to approximate a quantum state by sampling from a prior distribution and querying a quantum channel. The latter task aims to extract uniform and independent randomness against quantum adversaries. For both tasks, we use trace distance to measure the closeness between the processed state and the ideal target state.
We show that the minimal amount of samples for achieving an $\varepsilon$-covering is given by the $(1-\varepsilon)$-hypothesis testing information (with additional logarithmic additive terms), while the maximal extractable randomness for an $\varepsilon$-secret extractor is characterized by the conditional $(1-\varepsilon)$-hypothesis testing entropy.

When performing independent and identical repetitions of the tasks, our one-shot characterizations lead to tight asymptotic expansions of the above-mentioned operational quantities. We establish their second-order rates given by the quantum mutual information variance and the quantum conditional information variance, respectively.
Moreover, our results extend to the moderate deviation regime, which are the optimal asymptotic rates when the trace distances vanish at sub-exponential speed.
Our proof technique is direct analysis of trace distance without smoothing.

\end{abstract}
\maketitle

\section{Introduction} \label{sec:introduction}

Questing the optimal rates for information-processing tasks is a core problem in classical and quantum information theory \cite{Sha48, Wol57, Hol73, SW97, Hol98, ON99, Win99, Win99b, HN03}. Nowadays, considerable research focus has shifted from the first-order characterization of the optimal rates to the \emph{second-order} quantities in the asymptotic expansions of the optimal rates in coding blocklengths \cite{Str62, Hay09b, PPV10, Tan14, Li14, TH13, WWY14, TV15, DTW16, Tom16, TBR16}.
Such second-order terms quantifying how much extra cost one has to pay in non-asymptotic scenarios are of significant importance both in theory and practice.

Indeed, much progress has been made in deriving the exact second-order rates for numerous quantum information-theoretic protocols. Nevertheless, this problem remains open for certain tasks such as \emph{privacy amplification against quantum side information} (or called \emph{randomness extraction}) \cite{Ren05, TH13}, where the operational quantities usually being used as a security criterion is the trace distance \cite{Can01, Ren05, RK05, Unr10, Tom16, ABJ+20, Dup21, SGC22a}. This manifests the fact that existing analysis on operational quantities in terms of the trace distance as in various quantum information-theoretic tasks \cite{NC09, Wilde2, Tom16, Hay17, Wat18, KW20} still has room for improvement. Hence, this problem will be the main focus of this work. Moreover, we hope the proposed analysis on the trace distance would shed new lights on the \emph{one-shot quantum information theory} \cite{Tom16}.

In this paper, we study two tasks. The first task is privacy amplification against quantum side information \cite{Ren05, TSS+11, Hay13, PR14, Hay16, SIAM17, Tom16, Dup21}.
Suppose that a classical source $X$ at Alice's disposal may be correlated with a quantum adversary Eve, which can be modelled as a joint  classical-quantum (c-q) state $\rho_{XE}$. The goal of Alice is to extract from $X$  as much uniform randomness as possible that is independent of Eve. Due to operational motivation of composability \cite{Can01, Ren05, Unr10, Hay13, Tom16}, the \emph{trace distance} is usually adopted as the security criterion to measure how far the extracted state is from the perfectly uniform and independent randomness, i.e.~
\begin{align}
\Delta(X{\,|\,}E)_\rho :=    \frac12 \mathds{E}_h \left\| \mathcal{R}^h\left( \rho_{XE} \right) - \frac{\mathds{1}_Z}{|\mathcal{Z}|} \otimes \rho_E \right\|_1,
\end{align}
where $\mathcal{R}^h$ denotes the random hash function applied by Alice, and $\|\cdot\|_1$ is the Schatten $1$-norm.
A randomness extractor satisfying $\Delta(X{\,|\,}E)_\rho\leq \varepsilon$ is then said to be \emph{$\varepsilon$-secret}.
We define $\ell^\varepsilon(X{\,|\,}E)_\rho$ as the \emph{maximal extractable randomness} $|\mathcal{Z}|$ for $\varepsilon$-secret randomness extractors.

The second task studied in this work is \emph{quantum soft covering} \cite{CG22}. Consider a c-q state $\rho_{XB}:= \sum_{x\in\mathcal{X}} p_X(x)|x\rangle \langle x|\otimes \rho_B^x$. 
The goal of quantum soft covering is to approximate the marginal state $\rho_B$ using a random codebook $\mathcal{C}$ with certain size $|\mathcal{C}|$.
Here, the codewords in $\mathcal{C}$ are independently sampled from the distribution $p_X$; through the c-q channel $x\mapsto \rho_B^x$, one may construct an approximation state $\frac{1}{|\mathcal{C}|}\sum_{x\in\mathcal{C}} \rho_B^x$ to accomplish the goal.
Again, the trace distance is used as the figure of merit to measure closeness, i.e.~
\begin{align}
    \Delta(X{\,:\,}B)_\rho := \frac12 \mathds{E}_{\mathcal{C}} \left\| \frac{1}{|\mathcal{C}|}\sum_{x\in\mathcal{C}} \rho_B^x - \rho_B \right\|_1.
\end{align}
We say that the random codebook $\mathcal{C}$ achieves an \emph{$\varepsilon$-covering} if it satisfies $\Delta(X{\,:\,}B)_\rho \leq \varepsilon$.
We then define the \emph{$\varepsilon$-covering number} $M^\varepsilon(X{\,:\,}E)_\rho$ as the \emph{minimum random codebook size} $|\mathcal{C}|$ to realize an \emph{$\varepsilon$-covering}.

Our main result is to provide \emph{one-shot characterizations} for both the operational quantities $\ell^\varepsilon(X{\,|\,}E)_\rho$ and $M^\varepsilon(X{\,:\,}E)_\rho$.
For privacy amplification, we show that for every $0<\varepsilon<1$, the maximal extractable randomness using strongly $2$-universal hash functions is characterized by the \emph{$(1-\varepsilon)$-conditional hypothesis testing entropy} $H_\textnormal{h}^{1-\varepsilon}(X{\,|\,}E)_\rho$ \cite{TH13}:
\begin{align}
    \log \ell^\varepsilon(X{\,|\,}E)_\rho \approx H_\textnormal{h}^{1-\varepsilon\pm \delta}(X{\,|\,}E)_\rho\,.
\end{align}
Here, ``$\approx$'' means equality up to some logarithmic additive terms, and $\delta$ is a parameter that be chosen for optimization (see Theorem~\ref{theo:one-shot_second-order_PA} for the precise statement, and see Section~\ref{sec:notation} for detailed definitions).
With a similar flavor, we prove that for every $0<\varepsilon<1$, the minimal random codebook size for quantum soft covering is characterized by the \emph{$(1-\varepsilon)$-conditional hypothesis testing information}: $I_\textnormal{h}^{1-\varepsilon}(X{\,:\,}B)_\rho$  \cite{MW14}:
\begin{align}
    \log M^\varepsilon(X{\,:\,}B)_\rho \approx I_\textnormal{h}^{1-\varepsilon\pm \delta}(X{\,:\,}B)_\rho
\end{align}
(see Theorem~\ref{theo:one-shot_second-order_covering} for the precise statement).
Contrary to the previous studies, the established one-shot entropic characterizations established in this work are not smoothed min- and max-entropies \cite{Ren05, KRS09, CBR14, TCR10, TH13, Tsurumaru2021EquivalenceOT}.

In the scenario that the underlying states are identical and independently distributed, the established one-shot characterizations lead to the following second-order asymptotic expansions of the optimal rates, respectively (Propositions~\ref{prop:second-order_PA} and \ref{prop:second-order_covering}):
\begin{align}
    \log \ell^{\varepsilon}(X^n {\,|\,} E^n)_{\rho^{\otimes n}} &= nH(X{\,|\,}E)_{\rho} + \sqrt{nV(X{\,|\,}E)_{\rho}}\,\Phi^{-1}(\varepsilon) + O\left(\log n\right); \\
    \log M^\varepsilon(X^n{\,:\,}B^n)_{\rho^{\otimes n}} &= n I(X{\,:\,}B)_\rho - \sqrt{n V(X{\,:\,}B) }\, \Phi^{-1}(\varepsilon) + O\left(\log n\right).
\end{align}
Here, the first-order terms are the conditional quantum entropy $H(X{\,|\,}E)_{\rho}$ and the quantum mutual information $I(X{\,:\,}B)_\rho$, whereas the second-order rates that can be expressed as the quantum conditional information variance $V(X{\,|\,}E)$ and the quantum mutual information variance $V(X{\,:\,}B)$; and $\Phi^{-1}$ is the inverse of the cumulative normal distribution.

Furthermore, our results extend to the \emph{moderate deviation regime} \cite{CH17, CTT2017}. Namely, we derive both the optimal rates when the trace distances approach zero sub-exponentially (Propositions~\ref{prop:moderate_PA} and \ref{prop:moderate_covering}):
\begin{align}
    \frac1n \log \ell^{\varepsilon_n} (X^n{\,|\,}E^n)_{\rho^{\otimes n}} &= H(X{\,|\,}E)_\rho - \sqrt{2 V(X{\,|\,}E) }\, a_n + o\left(a_n\right); \\
    \frac1n \log M^{\varepsilon_n} (X^n{\,:\,}B^n)_{\rho^{\otimes n}} &= I(X{\,:\,}B)_\rho + \sqrt{2 V(X{\,:\,}B) }\, a_n + o\left(a_n\right).
\end{align}
Here, $(a_n)_{n\in\mathds{N}}$ is any moderate sequence satisfying $a_n\to 0$ and $n a_n^2 \to \infty$; and $\varepsilon_n := \mathrm{e}^{-n a_n^2}\to 0$\,.

\medskip
The rest of the paper is organized as follows.
In Section~\ref{sec:notation}, we introduce notations and auxiliary lemmas that will be used in our derivations.
Section~\ref{sec:PA} presents the one-shot and second-order characterizations of privacy amplification against quantum side information and 
Section~\ref{sec:covering} is devoted to quantum soft covering and its one-shot and second-order characterizations. We conclude our paper in Section~\ref{sec:conclusions}.

\section{Notation and Auxiliary Lemmas} \label{sec:notation}
For an integer $M\in\mathds{N}$, we denote $[M]:= \{1,\ldots, M\}$. We use `$\wedge$' to indicate `minimum value' between two scalars or the conjunction `and' between two statements. 
We use $\mathds{1}_{A}$ to denote the indicator function for a condition $A$.
The density operators considered in this paper are positive semi-definite operators with unit trace. 
For a trace-class operator $H$, the trace class norm (also called Schatten-$1$ norm) is defined by
\begin{align}
    \|H\|_1 := \Tr\left[ \sqrt{ H^\dagger H } \right].
\end{align}
For positive semi-definite operator $K$ and positive operator $L$, we use the following short notation for the \emph{noncommutative quotient}.
\begin{align} \label{eq:quotient}
\frac{K}{L}:= L^{-\frac{1}{2}} K L^{-\frac{1}{2}}\,.
\end{align} 
We use $\mathds{E}_{x\sim p_X}$ to stand for taking expectation where the underlying random variable is $x$ with probability distribution $p_X$ (with finite support), e.g.~$\rho_{XB} = \sum_{x\in\mathcal{X}} p_X(x)|x\rangle \langle x| \otimes \rho_B^x \equiv \mathds{E}_{x\sim p_X}\left( |x\rangle \langle x| \otimes \rho_B^x\right)$.

For density operators $\rho$ and $\sigma$, we define the $\varepsilon$-\emph{information spectrum divergence} \cite{HN03, NH07} as
\begin{align}
	D_\text{s}^{\varepsilon}(\rho \parallel \sigma) &:= \sup_{c\in\mathds{R}} \left\{ \log c : \Tr\left[ \rho \left\{ \rho \leq c \sigma \right\} \right] \leq \varepsilon  \right\}, \label{eq:Ds}
\end{align}
and the $\varepsilon$-\emph{hypothesis testing divergence} \cite{TH13, WR13, Li14} as
\begin{align} \label{eq:Dh}
    D_\text{h}^\varepsilon(\rho \,\|\, \sigma) := \sup_{0\leq T\leq \mathds{1} } \left\{ -\log \Tr[\sigma T] : \Tr[\rho T] \geq 1 - \varepsilon \right\}.
\end{align}
For positive semi-definite operator $\rho$ and positive definite operator $\sigma$, we define the \emph{collision divergence}\footnote{Note that for unnormalized $\rho$, the collision divergence is usually defined as $\log \Tr[ ( \sigma^{-\sfrac{1}{4}} \rho  \sigma^{-\sfrac{1}{4}} )^2 ] - \log \Tr[\rho]$. However, we do not use such a definition for notational convenience. 
Indeed, our derivations will rely on the joint convexity of $\exp D_2^*$ (see Lemma~\ref{lemm:joint_convexity}) which holds for positive semi-definite operator $\rho$ and the definition given in \eqref{eq:sand} as well.} \cite{Ren05} as
\begin{align} \label{eq:sand}
    D^*_2(\rho \,\|\, \sigma) :=  \log \Tr\left[ \left( \sigma^{-\frac{1}{4}} \rho  \sigma^{-\frac{1}{4}} \right)^2 \right].
\end{align}
The \emph{quantum relative entropy} \cite{Ume62, HP91} and \emph{quantum relative entropy variance} \cite{TH13,Li14} for density operator $\rho$ and positive definite operator $\sigma$ are defined as
\begin{align}
    D(\rho \,\|\,\sigma) &:= \Tr\left[ \rho \left( \log \rho - \log \sigma \right) \right];\\
    V(\rho \,\|\,\sigma) &:= \Tr\left[ \rho \left( \log \rho - \log \sigma \right)^2 \right] - \left( D(\rho \,\|\,\sigma) \right)^2.
\end{align}

By a classical-quantum state $\rho_{XB} = \sum_{x\in\mathcal{X}} p_X(x) |x\rangle\langle x|\otimes \rho_B^x$, we mean $p_X$ is a probability mass function on $\mathcal{X}$ and each $\rho_B^x$ is a density operator on system $B$.
We define the $\varepsilon$-\emph{hypothesis testing information} \cite{MW14} and the \emph{conditional $\varepsilon$-hypothesis testing entropy} \cite{TH13}\footnote{Note that we did not optimize the second argument but instead just put $\rho_B$; hence our definitions are slightly different from that proposed in Refs.~\cite{TH13, MW14}}, respectively, as
\begin{align}  \label{eq:one-shot_entropy}
    I_\text{h}^\varepsilon(X{\,:\,}B)_\rho &:= D_\text{h}^\varepsilon\left(\rho_{XB} \,\|\,\rho_X\otimes \rho_B\right); \quad
    H_\text{h}^\varepsilon(X{\,|\,}B)_\rho := -D_\text{h}^\varepsilon\left(\rho_{XB} \,\|\,\mathds{1}_X\otimes \rho_B\right).
\end{align}
We note that both the quantities $I_\text{h}^\varepsilon(X{\,:\,}B)_\rho$ and $H_\text{h}^\varepsilon(X{\,|\,}B)_\rho$ can be formulated as a semi-definite optimization problem \cite{TH13}, \cite[\S 1]{Wat18}.
The \emph{quantum mutual information}, \emph{quantum conditional entropy}, \emph{quantum mutual information variance}, and the \emph{quantum conditional information variance} of $\rho_{XB}$ are defined, respectively as
\begin{align}
    I(X{\,:\,}B)_\rho &:= D(\rho_{XB} \,\|\,\rho_X\otimes \rho_B)\,; \quad 
    H(X{\,|\,} B)_\rho := -D(\rho_{XB} \,\|\,\mathds{1}_X\otimes \rho_B)\,; \\
    V(X{\,:\,}B)_\rho &:= V(\rho_{XB} \,\|\,\rho_X\otimes \rho_B)\,; \quad
    V(X{\,|\,} B)_\rho := V(\rho_{XB} \,\|\,\mathds{1}_X\otimes \rho_B)\,.
\end{align}

For any self-adjoint operator $H$ with eigenvalue decomposition $H = \sum_i \lambda_i |e_i\rangle \langle e_i|$,
 we define the set $\textnormal{\texttt{spec}}(H):= \{\lambda_i\}_i$ to be the eigenvalues of $H$, and $|\textnormal{\texttt{spec}}(H)|$ to be the number of distinct eigenvalues of $H$.
We define the \emph{pinching map} with respect to $H$ as
\begin{align}
	\mathscr{P}_H[L] : L\mapsto \sum_{i=1} e_i L e_i\, ,
\end{align}
where $e_i$ is the spectrum projection onto $i$th distinct eigenvalues.
\begin{lemm}[{Pinching inequality} {\cite{Hay02}}] \label{lemm:pinching}
For every $d$-dimensional self-adjoint operator $H$ and positive semi-definite operator $L$,
\begin{align} \label{eq:pinching_inequality}
		\mathscr{P}_H[L] \geq \frac{1}{|\textnormal{\texttt{spec}}(H)|} L\,.
\end{align}
Moreover, it holds that for every $n\in\mathds{N}$, 
\begin{align}
\left|\textnormal{\texttt{spec}}\left(H^{\otimes n}\right)\right| \leq (n+1)^{d-1}\,.
\end{align}
\end{lemm}

We have the following relation between divergences that will be used in our proofs.
\begin{lemm}[Relation between divergences {\cite[Lemma~12, Proposition~13, Theorem~14]{TH13}}] \label{lemm:relation}
For every density operator $\rho$, positive semi-definite operator $\sigma$, $0< \varepsilon < 1$, and $0 < \delta < 1-\varepsilon$, we have
\begin{align}
    D_\textnormal{s}^{\varepsilon - \delta}\left( \mathscr{P}_{\sigma}[\rho] \,\|\,\sigma \right) &\leq D_\textnormal{h}^{\varepsilon + \delta}(\rho \,\|\,\sigma) + \log |\textnormal{\texttt{spec}}(\sigma)| + 2 \log \delta\,; \\
    D_\textnormal{h}^\varepsilon(\rho \,\|\,\sigma) &\leq D_\textnormal{s}^{\varepsilon + \delta} (\rho \,\|\,\sigma) - \log \delta\,.
\end{align}
\end{lemm}

\begin{lemm}[Lower bound on the collision divergence {\cite[Theorem~3]{BEI17}}] \label{lemm:D2toDs}
For every $0<\eta<1$ and $\lambda_1, \lambda_2>0$, density operator $\rho$, and positive semi-definite $\sigma$, we have\footnote{Ref.~\cite[Thm.~3]{BEI17} is stated for $\lambda_1= \lambda \in (0,1)$ and $\lambda_2 = 1-\lambda$. We remark that the result applies to the case $\lambda_1, \lambda_2>0$ by following the same proof.}
\begin{align}
    \exp D_2^*\left( \rho  \,\|\, \lambda_1 \rho + \lambda_2 \sigma \right) \geq \frac{1-\eta}{\lambda_1 + \lambda_2 \cdot \mathrm{e}^{ - D_\text{s}^\eta (\rho  \,\|\, \sigma )}  }\,.
\end{align}
\end{lemm}

\begin{lemm}[Joint convexity {\cite[Proposition 3]{FL13}, \cite{MDS+13, WWY14}}] \label{lemm:joint_convexity}
The map
\begin{align}
    (\rho,\sigma) \mapsto \exp D_2^*(\rho\,\|\,\sigma)
\end{align}
is jointly convex on all positive semi-definite operators $\rho$ and positive definite operators $\sigma$.
\end{lemm}

\begin{lemm}[Variational formula of the trace distance \cite{Hel67, Hol72}, {\cite[\S 9]{NC09}}]\label{lemm:1norm}
For density operators $\rho$ and $\sigma$,
\begin{align}
\frac{1}{2}\left\|\rho- \sigma\right\|_1 = \sup_{0\leq \Pi \leq \mathds{1}} \Tr\left[\Pi(\rho-\sigma)\right].
\end{align}
\end{lemm}
We list some basic properties of the noncommutative quotient introduced in \eqref{eq:quotient} as follows. Since they follow straightforwardly from basic matrix theory, we will use them in our derivations without proofs.
\begin{lemm}[Properties of the noncommutative quotient] \label{lemm:quiotient}
The noncommutative quotient defined in \eqref{eq:quotient} satisfies the following:
\begin{enumerate}[(a)]
    \item\label{item:a} $\displaystyle \frac{A}{B} \geq 0$ for all $A\geq 0$ and $B>0$;
    \item\label{item:b} $\displaystyle\frac{A+B}{C} = \frac{A}{C} + \frac{B}{C}$ for all $A,B\geq 0$ and $C>0$;
    \item\label{item:c} $\displaystyle\frac{A}{A+B} \leq \mathds{1}$ for all $A\geq 0$ and $B>0$;
    \item\label{item:d} $\displaystyle\Tr\left[ A \frac{B}{C} \right] = \Tr\left[ B \frac{A}{C}\right]$ for all $A,B\geq 0$ and $C>0$. 
    
    \item\label{item:e} $\displaystyle\Tr\left[ A \frac{A}{B} \right] = \exp D_2^*(A\,\|\,B)$ for all $A\geq 0$ and $B>0$.
\end{enumerate}
\end{lemm}

\begin{lemm}[Jensen's inequalities {\cite{Cho74},\cite{HP03}}] \label{lemm:Jensen}
Let $\Phi$ be a unital positive linear map between two spaces of bounded operators (possibly with different dimensions), and let $f$ be an operator concave function. Then for every  self-adjoint operator $H$, one has
\begin{align}
    \Phi(f(H)) \leq f(\Phi(H))\,.
\end{align}
For examples, the unital positive linear map $\Phi$ considered in this paper include 
(i) taking expectation $\mathds{E}$ on (possibly matrix-valued) random variables;
(ii) the pinching map $\mathscr{P}$;
and (iii) the functional $H\mapsto \Tr[\rho H] \in [0,1]$ for some density operator $\rho$.
\end{lemm}

\begin{lemm}[Second-order expansion \cite{TH13, Li14}]\label{lemm:second}
For every density operator $\rho$, positive definite operator $\sigma$, $0<\varepsilon<1$, and $\delta = O(\sfrac{1}{\sqrt{n}})$, we have the following expansion\footnote{
We note that Lemmas~\ref{lemm:second} and \ref{lemm:moderate} were originally stated for normalized $\sigma$ \cite{TH13, Li14, CTT2017}.
They hold for positive definite operator $\sigma$ as well by observing that
$D_\text{h}^\varepsilon(\rho\,\|\,\lambda\sigma) = D_\text{h}^\varepsilon(\rho\,\|\,\sigma) - \log \lambda$, $D(\rho\,\|\,\lambda\sigma) = D(\rho\,\|\,\sigma) - \log \lambda$, 
and $V(\rho\,\|\,\lambda\sigma) = V(\rho\,\|\,\sigma)$ for all $\lambda>0$.
}:
\begin{align}
    D_\textnormal{h}^{\varepsilon \pm \delta} \left(\rho^{\otimes n} \,\|\, \sigma^{\otimes n} \right) = n D\left(\rho \,\|\, \sigma\right) + \sqrt{n V\left(\rho \,\|\, \sigma\right)} \Phi^{-1} (\varepsilon) + O\left(\log n\right),
\end{align}
where $\Phi$ is the cumulative normal distribution $\Phi(u):= \int_{-\infty}^u \frac{1}{\sqrt{2\pi}} \mathrm{e}^{-\frac12 t^2}\mathrm{d}t$, and its inverse $\Phi^{-1}(\varepsilon) := \sup\{u\mid \Phi(u)\leq \varepsilon\}$.
\end{lemm}

\begin{lemm}[Moderate deviations {\cite[Theorem 1]{CTT2017}}] \label{lemm:moderate}
Let $(a_n)_{n\in\mathds{N}}$ be a moderate sequence satisfying 
\begin{align} \label{eq:an}
        \lim_{n\to \infty} a_n = 0 \ , \ 
        \lim_{n\to \infty} n a_n^2 = \infty \pl. 
\end{align}
and let $\varepsilon_n := \mathrm{e}^{-n a_n^2}$. For any density operator $\rho$ and positive definite operator $\sigma$, the following asymptotic expansions hold \footnote{
We note that the $\varepsilon$-hypothesis testing divergence $D_\textnormal{h}^{ \varepsilon }$ used in \cite{CTT2017} has an additional $\log (1-\varepsilon)$ term than our definition in \eqref{eq:Dh}.
Nonetheless, this does not affect the moderate deviation results since the additional terms are $\frac1n \log (1-\varepsilon_n) = o(a_n)$ and $\frac1n \log \varepsilon_n = - a_n^2 = o(a_n)$ for any moderate sequence $(a_n)_{n\in\mathds{N}}$.
}:
\begin{align}
\begin{dcases}
    \frac{1}{n} D_\textnormal{h}^{1 - \varepsilon_n } \left(\rho^{\otimes n} \,\|\, \sigma^{\otimes n} \right) = D\left(\rho \,\|\, \sigma\right) + \sqrt{2 V\left(\rho \,\|\, \sigma\right)} a_n + o\left(a_n\right); \\
    \frac{1}{n} D_\textnormal{h}^{\varepsilon_n } \left(\rho^{\otimes n} \,\|\, \sigma^{\otimes n} \right) = D\left(\rho \,\|\, \sigma\right) - \sqrt{2 V\left(\rho \,\|\, \sigma\right)} a_n + o\left(a_n\right). \\
\end{dcases}
\end{align}
\end{lemm}


\section{Privacy Amplification against Quantum Side Information} \label{sec:PA}
Let $\mathcal{H}_E$ is finite dimensional be a finite dimensional Hilbert space. Consider a classical-quantum state $\rho_{XE} = \sum_{x\in\mathcal{X}} p_X(x) |x\rangle \langle x|\otimes \rho_E^x$. Without loss of generality, we assume that the marginal density $\rho_E$ is invertible. 

In this work, we consider a \emph{strongly $2$-universal random hash function} $h :\mathcal{X} \to \mathcal{Z}$ is a random function which satisfies for all $ x, x'\in \mathcal{X}$ with $x\neq x'$ and $z,z'\in \mathcal{Z}$,
	\begin{align}
		\Pr_{h}\left\{ h(x) = z \; \wedge \; h(x') = z' \right\} = \frac{1}{|\mathcal{Z}|^2}\, .
	\end{align}
Namely, the output $h(x)$ for each input $x$ is uniform and pairwise independent.
Alice applies the linear operation $\mathcal{R}^h_{X\to Z} $ on her system $X$ by the following:
\begin{align}
	\mathcal{R}^h(\rho_{XE}) &:= \sum_{x\in\mathcal{X}} p_X(x)|h(x)\rangle \langle h(x)|\otimes \rho_E^x \, .
\end{align}
A perfectly randomizing channel $\mathcal{U}_{X\to Z}$ from $\mathcal{X}$ to $\mathcal{Z}$ is defined as
\begin{align}
	\mathcal{U}(\theta_X) = \frac{\mathds{1}_Z}{|\mathcal{Z}|}\left(\sum_{x}\theta_X(x)\right).
\end{align}

We define the \emph{maximal extractable randomness} for an $\varepsilon$-secret extractor \cite[\S 7]{Tom16} as
\begin{align}
    \ell^{\varepsilon}(X{\,|\,}E)_{\rho} := \sup\left\{  \ell \in \mathds{N}: |\mathcal{Z}|\geq \ell \wedge  
    \frac{1}{2}\mathds{E}_{h}\left\|\mathcal{R}^h(\rho_{XE})-\mathcal{U}(\rho_{XE}) \right\|_1 \leq \varepsilon \right\}.\label{def:maxl}
\end{align}
The main result of this section is to prove the following one-shot characterization of the operational quantity $\ell^{\varepsilon}(X{\,|\,}E)_{\rho}$.
\begin{theo}\label{theo:one-shot_second-order_PA}
	Let $\rho_{XE}$ be a classical-quantum state, and let $h(x): \mathcal{X} \to \mathcal{Z}$ be a strongly $2$-universal hash function. 
	Then, for every $0<\varepsilon<1$ and $0< c < \delta< \frac{\varepsilon}{3}\wedge\frac{(1-\varepsilon)}{2}$, we have 
	\begin{align}
H_\textnormal{h}^{1-\varepsilon + 3\delta}(X{\,|\,}E)_{\rho}  -\log\frac{\nu^2}{\delta^4} \leq\log \ell^{\varepsilon}(X{\,|\,}E)_{\rho} \leq H_\textnormal{h}^{1-\varepsilon-2\delta}(X{\,|\,}E)_{\rho} + \log\left(\frac{1+c}{c\delta}\right) + \log\left(\frac{\varepsilon+c}{\delta -c} \right).
	\end{align}
Here, $H_\textnormal{h}^{\varepsilon}$ is the conditional $\varepsilon$-hypothesis testing entropy defined in \eqref{eq:one-shot_entropy} and
$\nu = |\textnormal{\texttt{spec}}(\rho_E)|$.
\end{theo}
In the i.i.d setting, the one-shot characterization from Theorem \ref{theo:one-shot_second-order_PA} combined with the second-order expansion of $H_\textnormal{h}^{\varepsilon}$ leads to the following second-order asymptotics of $\log \ell^{\varepsilon}(X^n{\,|\,}E^n)_{\rho^{\otimes n}}$.

\begin{prop}\label{prop:second-order_PA}Let $\epsilon \in (0,1)$.
For any strongly $2$-universal hash function $h^n: \mathcal{X}^n\to \mathcal{Z}^n$ and \\$\frac{1}{2}\mathbb{E}_{h^n} \left\|({\mathcal{R}^{h^n}}-\mathcal{U}^{\otimes n})(\rho_{XE}^{\otimes n})\right\|_1\leq \epsilon$, we have the asymptotic lower bound:
\begin{align}
\log \ell^{\varepsilon}(X^n {\,|\,} E^n)_{\rho^{\otimes n}} = nH(X{\,|\,}E)_{\rho} + \sqrt{nV(X{\,|\,}E)_{\rho}}\Phi^{-1}(\varepsilon) + O\left(\log n\right).
\end{align}
\end{prop}

\begin{proof}
    Since $|\textnormal{\texttt{spec}}(\rho_E)|\leq (n+1)^{|\mathcal{H}_E|-1}$ for $|\mathcal{H}_E|$ being the rank of $\rho_E$, the additive terms grow of order $O(\log n)$. We apply the established one-shot characterization, Theorem \ref{theo:one-shot_second-order_PA}, with $\delta = n^{-\sfrac{1}{2}}$ and $c=\frac12 \delta$, together with the second-order expansion of the conditional hypothesis testing entropy, Lemma \ref{lemm:second}, to arrive at the claim.
\end{proof}
{ 
Moreover, the one-shot characterization, can be extended to the moderate deviation regime \cite{CTT2017, CH17}; namely, we derive the optimal rate of the maximal required output dimension when the error approaches zero moderately quickly. 
\begin{prop}[Moderate deviations for privacy amplification] \label{prop:moderate_PA}
For every classical-quantum state $\rho_{XB}$ and any moderate sequence $(a_n)_{n\in\mathds{N}}$ satisfying \eqref{eq:an} and $\varepsilon_n :=  \mathrm{e}^{-na_n^2}$, we have 
\begin{align} 
\begin{dcases}
    \frac1n \log \ell^{\varepsilon_n} (X^n{\,|\,}E^n)_{\rho^{\otimes n}} = H(X{\,|\,}E)_\rho - \sqrt{2 V(X{\,|\,}E) } a_n + o\left(a_n\right);\\
    \frac1n \log \ell^{1-\varepsilon_n} (X^n{\,|\,}E^n)_{\rho^{\otimes n}} = H(X{\,|\,}E)_\rho + \sqrt{2 V(X{\,|\,}E) } a_n + o\left(a_n\right).
\end{dcases}
\end{align}
\end{prop}
\begin{proof}
We prove the expansion for $\ell^{\varepsilon_n} (X^n{\,|\,}E^n)$. For every moderate deviation sequence $(a_n)_{n\in\mathds{N}}$, we let $\delta_n = \frac{1}{4}\e^{-na_n^2}$ satisfying \eqref{eq:an} and let $c_n = \frac{1}{2}\delta_n$. Then, $\varepsilon_n +2\delta_n$ or $\varepsilon_n -3\delta_n$ can be viewed as another $e^{-n b_n^2}$ for another moderate deviation sequence $(b_n)_{n\in\mathds{N}}$; we have $b_n = a_n + (b_n-a_n) = o(a_n)$ (see e.g.~\cite[\S 7.2]{CHD+21}), $\frac{1}{n}\log \delta_n = -a_n^2=o(a_n)$ and $\frac{1}{n}\log|\textnormal{\texttt{spec}}(\rho_E^{\otimes n})| = O\left(\frac{\log n}{n}\right) = o(a_n)$.
Then, applying Theorem \ref{theo:one-shot_second-order_PA} with Lemma \ref{lemm:moderate} leads to our first claim of the moderate derivation for privacy amplification. The second line follows similarly.
\end{proof}
}

We prove the lower bound and upper bound of $\log \ell^\varepsilon(X:E)_\rho $ in Section \ref{sec:lower} and Section \ref{sec:upper}, respectively.

\subsection{Direct Bound}\label{sec:lower}
We prove the lower bound on $\log \ell^\varepsilon(X{\,|\,}E)_\rho $ here.
\begin{proof}[Proof of achievability in {Theorem}~\ref{theo:one-shot_second-order_PA}]

We first claim that for any $c>0$ and strongly 2-universal hash function $h:\mathcal{X}^n \to \mathcal{Z}^n$,
\begin{align}
\frac{1}{2}\mathds{E}_{h} \left\| \mathcal{R}^h(\rho_{XE})-\mathcal{U}(\rho_{XE}) \right\|_1 \leq \tr\left[\rho_{XE}\left\{ \mathscr{P}_{\rho_E}[\rho_{XE}] > c \mathds{1}_X\otimes \rho_E \right\}\right] + \sqrt{ c|\textnormal{\texttt{spec}}(\rho_E)||\mathcal{Z}| }\,.\label{eq:pmain}
\end{align}
Let $\delta \in (0,\varepsilon)$ and let
\begin{align}
c &= \exp\left\{D_\text{s}^{1-\varepsilon+ \delta}(\mathscr{P}_{\rho_E}[\rho_{XE}]  \,\|\,\mathds{1}_X\otimes\rho_E) + \xi\right\},
\end{align}
for some small $\xi >0$. Then, by definition of the information spectrum divergence, \eqref{eq:Ds}, we have
\begin{align}
\tr\left[\rho_{XE}\left\{\mathscr{P}_{\rho_E}[\rho_{XE}]> c\mathds{1}_X\otimes \rho_E\right\}\right] < \varepsilon - \delta\,. \label{eq:pmain0}
\end{align}
Choose
$|\mathcal{Z}| = \left\lfloor {\frac{\delta^{2}}{c|\texttt{spec}(\rho_E)|}} \right\rfloor$. Then, by \eqref{eq:pmain} and \eqref{eq:pmain0},
the $\varepsilon$-secret criterion is satisfied, i.e.~
\begin{align}
\frac{1}{2}\mathds{E}_{h} \left\| \mathcal{R}^h(\rho_{XE})-\mathcal{U}(\rho_{XE}) \right\|_1 \leq \varepsilon\,,
\end{align}
By the definition of \eqref{def:maxl}, we have the following lower bound on $\log \ell^{\varepsilon}(X{\,|\,}E)_{\rho}$:
\begin{align}
\log \ell^{\varepsilon}(X{\,|\,}E)_{\rho} &\geq -D_\text{s}^{1-\varepsilon+ \delta}(\mathscr{P}_{\rho_E}[\rho_{XE}]  \,\|\,\mathds{1}_X\otimes\rho_E) - \xi - \log|\texttt{spec}(\rho_E)| + 2\log \delta\\
& \geq H_h^{1-\varepsilon + 3\delta}(X{\,|\,}B)_{\rho} - \xi -2\log|\texttt{spec}(\rho_E)| + 4\log \delta,
\end{align}
where we have used Lemma \ref{lemm:relation} in the last inequality. Since $\xi>0$ is arbitrary, taking $\xi\to 0$ gives our claim of lower bound in Theorem~\ref{theo:one-shot_second-order_PA}. 

\medskip
Now we move on to prove \eqref{eq:pmain}.
We first consider case where $\rho^x_E$ is invertible for each $x\in\mathcal{X}$ and later argue that the general case follows from approximation. 
Shorthand $p\equiv p_X$ and $\rho_x \equiv \rho_E^x$. For every $x\in\mathcal{X}$, we take the projection
\begin{align}
\Pi_x &= \left\{\mathscr{P}_{\rho_E}[p(x)\rho_x]\leq c\rho_E\right\}; \\
\Pi^\mathrm{c}_x&:= \mathds{1}_E-\Pi_x.
\end{align}
Observe that for every $z\in\mathcal{Z}$, we have 
\begin{align}
\mathds{E}_{h}\left[\sum_{x:h(x)=z}p(x)\rho_x\right] =  \frac{1}{|\mathcal{Z}|}\rho_E.
\end{align}
We then use the fact that the Schatten $1$-norm $\|\cdot\|_1$ is additive for direct sums to calculate
\begin{align}
&\frac{1}{2}\mathds{E}_{h}\left\|\mathcal{R}^h(\rho_{XE})-\mathcal{U}(\rho_{XE}) \right\|_1
\\&=  
\frac{1}{2}\mathds{E}_{h} \left\|\sum_z\ket{z}\bra{z}\otimes\left(\sum_{x:h(x)=z}p(x)\rho_x\right)- \frac{1}{|\mathcal{Z}|}\sum_{z}\ket{z}\bra{z}\otimes\rho_E\right\|_1 \\
& = \sum_{z\in \mathcal{Z}} \frac{1}{2}\mathds{E}_{h} \left\|\sum_{x:h(x)=z}p(x)\rho_x- \frac{1}{|\mathcal{Z}|}\rho_E\right\|_1\\
& = \sum_{z\in \mathcal{Z}} \frac{1}{2}\mathds{E}_{h} \left\|\sum_{x:h(x)=z}p(x)\rho_x- \mathds{E}_{h}\left[\sum_{x:h(x)=z}p(x)\rho_x\right]\right\|_1\\
&= \sum_{z\in \mathcal{Z}} \frac{1}{2}\mathds{E}_{h} \left\|\sum_{x:h(x)=z}p(x)\rho_x(\Pi^{\mathrm{c}}_x+\Pi_x)- \mathds{E}_h\left[\sum_{x:h(x)=z}p(x)\rho_x(\Pi^{\mathrm{c}}_x+\Pi_x)\right]\right\|_1\\
& \overset{(a)}{\leq} \sum_{z\in \mathcal{Z}} \frac{1}{2}\mathds{E}_{h} \left\|\sum_{x:h(x)=z} p(x)\rho_x\Pi^{\mathrm{c}}_x- \mathds{E}_h\left[\sum_{x:h(x)=z} p(x)\rho_x\Pi^{\mathrm{c}}_x\right]\right\|_1\\
&\quad + \sum_{z\in \mathcal{Z}} \frac{1}{2}\mathds{E}_{h} \left\|\sum_{x:h(x)=z} p(x)\rho_x\Pi_x- \mathds{E}_h\left[\sum_{x:h(x)=z}  p(x)\rho_x\Pi_x\right]\right\|_1\\
\begin{split}
&\leq  \sum_{z\in \mathcal{Z}} \frac{1}{2}\mathds{E}_{h} \left\|\sum_{x:h(x)=z} \Pi^{\mathrm{c}}_xp(x)\rho_x\Pi^{\mathrm{c}}_x- \mathds{E}_h\left[\sum_{x:h(x)=z} \Pi^{\mathrm{c}}_x p(x)\rho_x\Pi^{\mathrm{c}}_x\right]\right\|_1\\
&\quad + \sum_{z\in \mathcal{Z}} \frac{1}{2}\mathds{E}_{h} \left\|\sum_{x:h(x)=z} \Pi_xp(x)\rho_x\Pi^{\mathrm{c}}_x- \mathds{E}_h\left[\sum_{x:h(x)=z} \Pi_x p(x)\rho_x\Pi^{\mathrm{c}}_x\right]\right\|_1\\
& \quad + \sum_{z\in \mathcal{Z}} \frac{1}{2}\mathds{E}_{h} \left\|\sum_{x:h(x)=z} p(x)\rho_x\Pi_x- \mathds{E}_h\left[\sum_{x:h(x)=z} p(x)\rho_x\Pi_x\right]\right\|_1, 
\end{split} \label{eq:123}
\end{align}
where (a) follows from the triangle inequality of the Schatten $1$-norm $\|\cdot\|_1$, and we then use triangle inequality again to decompose the first term at (a) to arrive at \eqref{eq:123}.

We bound the three terms of \eqref{eq:123} as follows.
The first term of \eqref{eq:123} is bounded by
\begin{align}
&\sum_{z\in \mathcal{Z}} \frac{1}{2}\mathds{E}_{h} \left\|\sum_{x:h(x)=z}\Pi^{\mathrm{c}}_xp(x)\rho_x\Pi^{\mathrm{c}}_x- \mathds{E}_h\left[\sum_{ x:h(x)=z }\Pi^{\mathrm{c}}_x p(x)\rho_x\Pi^{\mathrm{c}}_x\right]\right\|_1\\
&\leq \sum_{z\in \mathcal{Z}} \mathds{E}_{h} \left\|\sum_{x:h(x)=z}\Pi^{\mathrm{c}}_x\left(p(x)\rho_x\right)\Pi^{\mathrm{c}}_x\right\|_1\\
& = \sum_{z\in \mathcal{Z}} \mathds{E}_{h} \tr\left[\sum_{x\in \mathcal{X}}\mathbf{1}_{h(x)=z}\ket{x}\bra{x}\otimes \left(p(x)\rho_x\{\mathscr{P}_{\rho_E}[p(x)\rho_x]>c\rho_E\}\right)\right]\\
& =  \mathds{E}_{h} \tr\left[\sum_{x\in \mathcal{X}} \ket{x}\bra{x}\otimes \left(p(x)\rho_x\{\mathscr{P}_{\rho_E}[p(x)\rho_x]>c\rho_E\}\right)\right]\\
&= \tr\left[\left(\sum_{x\in\mathcal{X}}\ket{x}\bra{x}\otimes (p(x)\rho_x)\right)\left(\sum_{x\in\mathcal{X}}\ket{x}\bra{x}\otimes \{\mathscr{P}_{\rho_E}[p(x)\rho_x]>c\rho_E\}\right)\right]\\
&= \tr\left[\rho_{XE}\{\mathscr{P}_{\rho_E}[\rho_{XE}]>c\mathds{1}_X\otimes\rho_E\}\right]. \label{eq:123_1}
\end{align}
Let us shorthand for notational convenience
\begin{align}
    H_{h,z} := \sum_{x:h(x)=z}\Pi_xp(x)\rho_x\Pi^{\mathrm{c}}_x- \mathds{E}_h\left[\sum_{x:h(x)=z}\Pi_x p(x)\rho_x\Pi^{\mathrm{c}}_x\right].
\end{align}
The second term of \eqref{eq:123} can be bounded by:
\begin{align}
\sum_{z\in \mathcal{Z}} \frac{1}{2}\mathds{E}_{h} \left\|H_{h,z}\right\|_1
&= \sum_{z\in \mathcal{Z}} \frac{1}{2}\mathds{E}_h \tr \sqrt{ H_{h,z}^\dagger H_{h,z}}\\
&\overset{(a)}{\leq} \sum_{z\in \mathcal{Z}} \frac{1}{2} \tr \sqrt{\mathds{E}_h \left[ H_{h,z}^{\dagger} H_{h,z}\right] }\\
&\overset{(b)}{=} \sum_{z\in \mathcal{Z}} \frac{1}{2}\Tr \sqrt{ \mathrm{Var}_{h} \left[ \sum_{x\in\mathcal{X}} \mathbf{1}_{\{h(x)=z\}} \Pi^{\mathrm{c}}_xp(x)\rho_x\Pi_x \right] }\\
&\overset{(c)}{=} \sum_{z\in \mathcal{Z}} \frac{1}{2}\Tr \sqrt{ \sum_{x\in\mathcal{X}}\mathrm{Var}_{h} \left[ \mathbf{1}_{\{h(x)=z\}} \Pi^{\mathrm{c}}_xp(x)\rho_x\Pi_x \right] }\\
&\overset{(d)}{\leq} \sum_{z\in \mathcal{Z}} \frac{1}{2}\Tr \sqrt{ \sum_{x\in\mathcal{X}}\mathds{E}_{h} \left[ \mathbf{1}_{\{h(x)=z\}} \Pi_x p(x)\rho_x \Pi^{\mathrm{c}}_x p(x)\rho_x \Pi_x \right] }\\
&\overset{(e)}{\leq} \sum_{z\in \mathcal{Z}} \frac{1}{2}\Tr \sqrt{ \sum_{x\in\mathcal{X}}\mathds{E}_{h} \left[ \mathbf{1}_{\{h(x)=z\}} \Pi_x(p(x)\rho_x)^2\Pi_x \right] } \\
&\overset{(f)}{\leq} \sum_{z\in \mathcal{Z}} \frac{1}{2}\Tr \sqrt{ \sum_{x\in\mathcal{X}} \frac{1}{|\mathcal{Z}|} \Pi_x(p(x)\rho_x)^2\Pi_x  }. \label{eq:123_2}
\end{align}
Here, in (a) we applied Jensen's inequality, Lemma~\ref{lemm:Jensen}, with expectation $\mathds{E}_h$ and the operator concavity of square-root;
	in (b) we denoted a \emph{matrix-valued variance} for a random matrix $H_h$ as $\mathrm{Var}[H_{h}]:= \mathds{E}[H_{h}^\dagger H_{h}] - (\mathds{E}[H_{h}])^\dagger \mathds{E}[H_{h}]$ (see e.g.~\cite[\S 2]{Tro15});
	in (c) we applied pairwise independent property of the strongly $2$-universal hash function (on each $x\in\mathcal{X}$);
	in (d) we used the operator monotone of square-root and $\mathrm{Var}[H_{h}] \leq \mathds{E}[H_{h}^\dagger H_{h}]$;
	(e) follows from $\Pi_x^\mathrm{c} \leq \mathds{1}_E$ and the operator monotonicity of square-root;
	in (f) we used the uniformity of the random hash function $h$.

The third term of \eqref{eq:123} can be bounded similar to the second term as:
\begin{align}
&\sum_{z\in \mathcal{Z}} \frac{1}{2}\mathds{E}_{h} \left\|\sum_{x\in\mathcal{X}}\mathbf{1}_{h(x)=z}p(x)\rho_x\Pi_x- \mathds{E}_h\left[\sum_{x\in\mathcal{X}}\mathbf{1}_{h(x)=z}p(x)\rho_x\Pi_x\right]\right\|_1\\
&\leq \sum_{z\in \mathcal{Z}} \frac{1}{2}\Tr \sqrt{ \sum_{x\in\mathcal{X}} \frac{1}{|\mathcal{Z}|} \Pi_x(p(x)\rho_x)^2\Pi_x  }.  \label{eq:123_3}
\end{align}

Now by \eqref{eq:123_2} and \eqref{eq:123_3}, we bound the sum of second and third term of \eqref{eq:123} as:
\begin{align}
&\sum_{z\in \mathcal{Z}} \Tr \sqrt{ \sum_{x\in\mathcal{X}}  \frac{1}{|\mathcal{Z}|} \Pi_x(p(x)\rho_x)^2\Pi_x }\\
&= \sum_{z\in \mathcal{Z}} \Tr \mathscr{P}_{\rho_E}\sqrt{ \sum_{x\in\mathcal{X}}  \frac{1}{|\mathcal{Z}|} \Pi_x(p(x)\rho_x)^2\Pi_x }\\
&\overset{(a)}{\leq} \sum_{z\in \mathcal{Z}} \Tr \sqrt{\mathscr{P}_{\rho_E} \left[\sum_{x\in\mathcal{X}} \frac{1}{|\mathcal{Z}|} \Pi_x(p(x)\rho_x)^2\Pi_x \right] }\\
&\overset{(b)}{=} \sum_{z\in \mathcal{Z}} \Tr\left[ \rho_E\sqrt{\mathscr{P}_{\rho_E}\left[ \sum_{x\in\mathcal{X}} \frac{1}{|\mathcal{Z}|}\Pi_x(p(x)\rho_x)^2\Pi_x \right] \rho_E^{-2}}\right]\\
&\overset{(c)}{\leq} \sum_{z\in \mathcal{Z}}  \sqrt{\Tr \left[ \sum_{x\in\mathcal{X}}  \frac{1}{|\mathcal{Z}|} \Pi_x(p(x)\rho_x)^2\Pi_x \rho_E^{-1} \right] }\\
& = \sum_{z\in \mathcal{Z}}  \sqrt{\sum_{x\in\mathcal{X}} \frac{1}{|\mathcal{Z}|} \Tr \left[ (p(x)\rho_x)^2\Pi_x  \rho_E^{-1}\Pi_x\right] }\\
& \overset{(d)}{\leq} \sum_{z\in \mathcal{Z}}  \sqrt{ \sum_{x\in\mathcal{X}} \frac{1}{|\mathcal{Z}|} \Tr \left[ (p(x)\rho_x)^2 c|\textnormal{\texttt{spec}}(\rho_E)|(p(x)\rho_x)^{-1}\right] }\\
& = |\mathcal{Z}|\sqrt{\frac{c}{|\mathcal{Z}|}|\textnormal{\texttt{spec}}(\rho_E)|}\\
& = \sqrt{c|\mathcal{Z}||\textnormal{\texttt{spec}}(\rho_E)|}. \label{eq:123_4}
\end{align}
where (a) follows from Jensen's inequality, Lemma~\ref{lemm:Jensen}, with the pinching map $\mathscr{P}_{\rho_E}$ and the operator concavity of square-root; 
(b) holds because now every term is commuting with $\rho_E$;
(c) follows from Jensen's inequality, Lemma~\ref{lemm:Jensen}, with the functional $\Tr[\rho_E (\, \cdot \,)]$ and the operator concavity of square-root;
in (d) we use the fact that
 \begin{align}
 	\Pi_x \rho_E^{-1} \Pi_x &\leq c \, \Pi_x \left(\mathscr{P}_{\rho_E}[p(x)\rho_x]\right)^{-1} \Pi_x \\
 	&\leq c\left(\mathscr{P}_{\rho_E}[p(x)\rho_x]\right)^{-1} \\
 	&\leq c \left|\textnormal{\texttt{spec}}(\rho_E)\right| (p(x)\rho_x)^{-1},  \label{eq:one-shot_second-order6}
 \end{align}
since 
\begin{align}
	\Pi_x = \left\{ \mathscr{P}_{\rho_E}[p(x)\rho_x] \leq c \rho_E \right\} = \left\{ \rho_E^{-1} \leq c \left(\mathscr{P}_{\rho_E}[p(x)\rho_x]\right)^{-1} \right\},
\end{align}
and we invoke the operator monotonicity of matrix inversion together with the pinching inequality (Lemma~\ref{lemm:pinching}), i.e.~\begin{align}
\mathscr{P}_{\rho_E}[p(x)\rho_x] \geq \frac{p(x)\rho_x }{|\texttt{spec}(\rho_E)|}.
\end{align}
The statement \eqref{eq:pmain} is proved by combining  \eqref{eq:123}, \eqref{eq:123_1}, \eqref{eq:123_2}, \eqref{eq:123_3}, and \eqref{eq:123_4}.

For the general non-invertible states $\{\rho_x\}_x$ we define the approximation
\begin{align}
\rho_x^\epsilon:= (1-\epsilon)\rho_x + \epsilon\frac{\mathds{1}_E}{|\mathcal{H}_E|},
\end{align}
and then 
	\[\rho_E^\epsilon := (1-\epsilon)\rho_E+\epsilon\frac{\mathds{1}_{E}}{|\mathcal{H}_E|}\ ,\  \rho_{XE}^\epsilon=(1-\epsilon)\rho^\epsilon_{XE}+\epsilon \rho_X \ten \frac{\mathds{1}_{E}}{|\mathcal{H}_E|}. \]
Moreover, since $\mathscr{P}_{\rho_E}$ is unital completely positive and trace-preserving, 
\begin{align}
\mathscr{P}_{\rho_E}\left[\rho_{XE}^\epsilon\right]:= (1-\epsilon)\mathscr{P}_{\rho_E}\left[\rho_{XE}\right] + \epsilon\rho_X\otimes\frac{\mathds{1}_E}{|\mathcal{H}_E|}.
\end{align}
It is clear that
\begin{align}
\lim_{\epsilon\to 0}\frac{1}{2}\mathds{E}_{h} \left\| \mathcal{R}^h(\rho_{XE}^\epsilon)-\mathcal{U}(\rho_{XE}^\epsilon) \right\|_1 = \frac{1}{2}\mathds{E}_{h} \left\| \mathcal{R}^h(\rho_{XE})-\mathcal{U}(\rho_{XE}) \right\|_1.
\end{align}
For the right-hand side of \eqref{eq:pmain}, we have
\begin{align}
\left\|\mathscr{P}_{\rho_E}\left[\rho_{XE}^\epsilon\right] - \mathscr{P}_{\rho_E}[\rho_{XE}]\right\|_1 \leq 2\epsilon
\end{align}
and for small enough $\epsilon$, the projection 
\begin{align}
\left\{ \mathscr{P}_{\rho_E}[(\rho_{XE})^{\epsilon}] > c \mathds{1}_X\otimes (\rho_E)^{\epsilon} \right\} = \left\{ \mathscr{P}_{\rho_E}[\rho_{XE}] > c \mathds{1}_X\otimes \rho_E \right\}.
\end{align}
In fact, because $\mathscr{P}_{\rho_E}[(\rho_{XE})]$ and $\mathds{1}_X\otimes \rho_E$ are commutative, they can be viewed as functions on the finite set of spectrum. Therefore, we have
\begin{align}
\lim_{\epsilon\to 0} \tr\left[(\rho_{XE})^{\epsilon}\left\{ \mathscr{P}_{\rho_E}[(\rho_{XE})^{\epsilon}] > c \mathds{1}_X\otimes (\rho_E)^{\epsilon} \right\}\right] = \tr\left[\rho_{XE}\left\{ \mathscr{P}_{\rho_E}[\rho_{XE}] > c \mathds{1}_X\otimes \rho_E \right\}\right],
\end{align}
which proves \eqref{eq:pmain} by approximation.
\end{proof}
\subsection{Converse Bound}\label{sec:upper}
We prove the upper bound on $\log \ell^\varepsilon(X{\,|\,}E)_\rho $ here.
\begin{proof}[Proof of converse of {Theorem}~\ref{theo:one-shot_second-order_PA}]
We denote $p \equiv p_X$. For every $0<c<\delta < \frac{1-\varepsilon}{2}$, and for any realization of the random hash function $h$, we choose the noncommutative quotient
\begin{align}
    \Pi = \frac{\mathcal{R}^h(\rho_{XE})}{\mathcal{R}^h(\rho_{XE}) + c^{-1}\cdot \mathcal{U}(\rho_{XE})}.
\end{align}
For every $\varepsilon$-secret randomness extractor, we calculate
\begin{align}
\varepsilon &\geq \frac{1}{2}\mathbb{E}_h \left\|(\mathcal{R}^h-\mathcal{U})(\rho_{XE})\right\|_1\\
&\overset{(a)}{\geq} \mathbb{E}_h \Tr\left[(\mathcal{R}^h-\mathcal{U})(\rho_{XE})\frac{\mathcal{R}^h(\rho_{XE})}{\mathcal{R}^h(\rho_{XE}) + c^{-1} \mathcal{U}(\rho_{XE})}\right]\\
&\overset{(b)}{\geq} \mathbb{E}_h \Tr\left[\mathcal{R}^h(\rho_{XE})\frac{\mathcal{R}^h(\rho_{XE})}{\mathcal{R}^h(\rho_{XE}) + c^{-1}\mathcal{U}(\rho_{XE})}\right]-c, \label{eq:Priv1}
\end{align}
where in (a) we used the lower bound on trace distance, Lemma \ref{lemm:1norm}, with $\Pi$, and in (b) we applied Lemma~\ref{lemm:quiotient}-\ref{item:d} to obtain the following estimation:
\begin{align*} 
\Tr\left[\mathcal{U}(\rho_{XE})\frac{\mathcal{R}^h(\rho_{XE})}{\mathcal{R}^h(\rho_{XE}) + c^{-1} \mathcal{U}(\rho_{XE})}\right] 
&= c\Tr\left[\mathcal{R}^h(\rho_{XE}) \frac{c^{-1}\mathcal{U}(\rho_{XE})}{\mathcal{R}^h(\rho_{XE}) + c^{-1} \mathcal{U}(\rho_{XE})}\right]\\ 
&\le c\Tr\left[\mathcal{R}^h(\rho_{XE})\right]\\
&= c\pl. 
\end{align*}
For every $z\in\mathcal{Z}$, we define
\begin{align}
\sigma_{zXE}:=& \ket{z}\bra{z}\otimes\left(\sum_{x:\, h(x)=z}\ket{x}\bra{x}\otimes p(x)\rho^x_E\right);\\
\tau_{zXE}:=& \left(\sum_{x:\, h(x) = z}\ket{x}\bra{x}\right) \otimes \left(\mathcal{R}^h(\rho_{XE})+c^{-1}\mathcal{U}(\rho_{XE}))\right).
\end{align}
Since $\mathcal{R}^h(\rho_{XE})= \sum_{z\in \mathcal{Z} }\sigma_{zE}$ for $\sigma_{zE} = \Tr_{X}[\sigma_{zXE}] := |z\rangle \langle z|\otimes \sum_{x:\, h(x)=z} p(x)\rho_E^x$,
we calculate
\begin{align}
\mathbb{E}_h \Tr\left[\mathcal{R}^h(\rho_{XE})\frac{\mathcal{R}^h(\rho_{XE})}{\mathcal{R}^h(\rho_{XE}) + c^{-1} \mathcal{U}(\rho_{XE})}\right] 
&=\sum_{z\in \mathcal{Z}} \mathbb{E}_h\Tr\left[\sigma_{zE}\frac{\mathcal{R}^h(\rho_{XE})}{\mathcal{R}^h(\rho_{XE})+ c^{-1}\mathcal{U}(\rho_{XE})}\right] \\
&\overset{(a)}{\geq}  \sum_{z\in \mathcal{Z}}\mathbb{E}_h\Tr\left[\sigma_{zE}\frac{\sigma_{zE}}{\mathcal{R}^h(\rho_{XE})+ c^{-1}\mathcal{U}(\rho_{XE})}\right]\\
&= \sum_{z\in \mathcal{Z}}\sum_{x:h(x)=z}\mathbb{E}_h\Tr\left[\ket{z}\bra{z}\ten p(x)\rho^x_E\frac{ \sigma_{zE} }{\mathcal{R}^h(\rho_{XE})+ c^{-1}\mathcal{U}(\rho_{XE})}\right]\\
&\overset{(b)}{\geq} \sum_{z\in \mathcal{Z}}\sum_{x:h(x)=z}\mathbb{E}_h\Tr\left[\ket{z}\bra{z}\ten p(x)\rho^x_E\frac{\ket{z}\bra{z}\ten p(x)\rho^x_E}{\mathcal{R}^h(\rho_{XE})+ c^{-1}\mathcal{U}(\rho_{XE})}\right]\\
&= \sum_{z\in \mathcal{Z}}\mathbb{E}_h \exp{D_2^*(\sigma_{zXE}  \,\|\,\tau_{zXE})} \\
&\overset{(c)}{\geq} \sum_{z\in \mathcal{Z}}\exp{D_2^*\left(\mathbb{E}_h \left[ \sigma_{zXE} \right]  \|\,\mathbb{E}_h \left[ \tau_{zXE}\right] \right)}, \label{eq:Priv2}
\end{align}
where 
in (a) we used $\mathcal{R}^h(\rho_{XE}) \geq \sigma_{zE}$ for all $z\in\mathcal{Z}$ and Lemma~\ref{lemm:quiotient}-\ref{item:a} \& \ref{item:b};
in (b) we used that for every $x, z$ s.t.~$h(x) = z$, $\sigma_{zE} \geq \ket{z}\bra{z}\otimes p(x)\rho^x_E$ and Lemma~\ref{lemm:quiotient} again;
and in (c) we used the joint convexity of $\exp D_2^*(\,\cdot\|\,\cdot)$, Lemma~\ref{lemm:joint_convexity}.
Then, we exploit the uniformity and independence of the random hash function to calculate the expectations:
\begin{align}
\mathbb{E}_h \left[ \sigma_{zXE} \right]
&= \frac{1}{|\mathcal{Z}|}\ket{z}\bra{z}\otimes\rho_{XE};\\
\mathbb{E}_h \left[ \left(\sum_{x:\, h(x) = z}\ket{x}\bra{x}\right)\otimes \mathcal{R}^h(\rho_{XE}) \right]
&= \mathds{E}_h \sum_{x,\bar{x},\bar{z}} |x\rangle \langle x| \otimes \mathbf{1}_{\{h(x)=z\}} \mathbf{1}_{\{h(\bar{x})=\bar{z}\}} |\bar{z}\rangle\langle\bar{z}| \otimes p(\bar{x}) \rho_E^{\bar{x}} \\
&= \sum_{x,\bar{x},\bar{z}} |x\rangle \langle x| \otimes \left[ \frac{1}{|\mathcal{Z}|} \mathbf{1}_{\{x=\bar{x}\}} \mathbf{1}_{\{z = \bar{z}\}} + \frac{1}{|\mathcal{Z}|^2} \mathbf{1}_{\{x\neq \bar{x}\}} \right] |\bar{z}\rangle\langle\bar{z}| \otimes p(\bar{x}) \rho_E^{\bar{x}} \\
&=\frac{1}{|\mathcal{Z}|}\ket{z}\bra{z}\otimes\rho_{XE} + \frac{1}{|\mathcal{Z}|^2} \mathds{1}_Z \otimes (\mathds{1}_X\otimes\rho_E-\rho_{XE});\\
\mathbb{E}_h \left[ \left(\sum_{x:\, h(x) = z}\ket{x}\bra{x}\right)\otimes c^{-1} \mathcal{U}(\rho_{XE}) \right]
&= c^{-1}\cdot\frac{\mathds{1}_Z}{|\mathcal{Z}|}\otimes\frac{\mathds{1}_X}{|\mathcal{Z}|}\otimes\rho_E.
\end{align}
Here, a crucial observation is that we can employ the direct-sum structure of $\mathds{1}_Z$ together with the definition of the collision diversion, \eqref{eq:sand}, to rewrite \eqref{eq:Priv2} as follows, i.e.~for every $z\in\mathcal{Z}$,
\begin{align}
&\exp{D_2^*\left(\mathbb{E}_h \left[ \sigma_{zXE} \right]  \,\|\,\mathbb{E}_h \left[ \tau_{zXE}\right] \right)} \\ \notag \\
    = &\exp D_2^* \left( \left. \frac{|z\rangle\langle z|}{|\mathcal{Z}|}\otimes \rho_{XE} \,\right\| \frac{|z\rangle\langle z|}{|\mathcal{Z}|}\otimes \rho_{XE} + \frac{\mathds{1}_Z}{|\mathcal{Z}|^2}  \otimes \left( (1+c^{-1}) \mathds{1}_X\otimes \rho_E - \rho_{XE}\right)    \right) \\
    =& \exp D_2^* \left( \left. \frac{|z\rangle\langle z|}{|\mathcal{Z}|}\otimes \rho_{XE} \,\right\| \frac{|z\rangle\langle z|}{|\mathcal{Z}|}\otimes \rho_{XE} + \frac{|z\rangle \langle z |}{|\mathcal{Z}|^2}  \otimes \left( (1+c^{-1}) \mathds{1}_X\otimes \rho_E - \rho_{XE}\right) \right) \\
    =& \exp D_2^* \left( \left. |\mathcal{Z}|^{-1}\rho_{XE} \,\right\|  |\mathcal{Z}|^{-1}\rho_{XE} + |\mathcal{Z}|^{-2}  \left( (1+c^{-1}) \mathds{1}_X\otimes \rho_E - \rho_{XE}\right) \right). \label{eq:Priv2_5}
\end{align}
Hence, \eqref{eq:Priv2} and \eqref{eq:Priv2_5} show that
\begin{align}
    &\sum_{z\in \mathcal{Z}}\exp{D_2^*\left(\mathbb{E}_h \left[ \sigma_{zXE} \right]  \|\,\mathbb{E}_h \left[ \tau_{zXE}\right] \right)}
    \\ =& \sum_{z\in \mathcal{Z}} \exp D_2^* \left( \left. |\mathcal{Z}|^{-1}\rho_{XE} \,\right\|  |\mathcal{Z}|^{-1}\rho_{XE} + |\mathcal{Z}|^{-2}  \left( (1+c^{-1}) \mathds{1}_X\otimes \rho_E - \rho_{XE}\right) \right) \\
    =& \exp D^*_2\left(\rho_{XE} \,\| \left(1-|\mathcal{Z}|^{-1}\right)\rho_{XE}+ |\mathcal{Z}|^{-1}(1+c^{-1}) \mathds{1}_X\ten \rho_E\right) \\
    {\geq}& (\delta + \varepsilon)\left(1-\frac{1}{|\mathcal{Z}|}+ \frac{1+c^{-1}}{|\mathcal{Z}|}\e^{-D_\text{s}^{1-\varepsilon-\delta}(\rho_{XE}  \,\|\,\mathds{1}_X\otimes\rho_E)}\right)^{-1},
    \label{eq:Priv3}
\end{align}
where we apply Lemma \ref{lemm:D2toDs} with $\eta = 1-\varepsilon-\delta$, $\lambda_1 = 1-\frac{1}{|\mathcal{Z}|}$, and $\lambda_2 = \frac{1}{|\mathcal{Z}|}(1+c^{-1})$ in the last inequality.

Combining \eqref{eq:Priv1}, \eqref{eq:Priv2}, and \eqref{eq:Priv3} gives
\begin{align}
\varepsilon \geq (\delta + \varepsilon)\left(1-\frac{1}{|\mathcal{Z}|}+ \frac{1+c^{-1}}{|\mathcal{Z}|}\e^{-D_\text{s}^{1-\varepsilon-\delta}(\rho_{XE}  \,\|\,\mathds{1}_X\otimes\rho_E)}\right)^{-1} - c,
\end{align}
which can be translated to
\begin{align}
\log|\mathcal{Z}| 
&\leq -D_\text{s}^{1-\varepsilon-\delta}(\rho_{XE}  \,\|\,\mathds{1}_X\otimes\rho_E) + \log(1+c^{-1}) -\log\left(\frac{\delta -c}{\varepsilon+c}+ \frac{1}{|\mathcal{Z}|}\right)\\
&\leq -D_\text{s}^{1-\varepsilon-\delta}(\rho_{XE}  \,\|\,\mathds{1}_X\otimes\rho_E) + \log(1+c^{-1}) -\log\left(\frac{\delta -c}{\varepsilon+c}\right)\\
&\overset{(a)}{\leq} H_\text{h}^{1-\varepsilon-2\delta}(X{\,|\,}E)_{\rho} + \log(1+c^{-1}) -\log\left(\frac{\delta -c}{\varepsilon+c}\right) -  \log \delta\\
&= H_\text{h}^{1-\varepsilon-2\delta}(X{\,|\,}E)_{\rho} + \log\left(\frac{1+c}{c\delta}\right) + \log\left(\frac{\varepsilon+c}{\delta -c}\right),
\end{align}
where we applied Lemma~\ref{lemm:relation} in (a).
That completes the proof.
\end{proof}

\section{Quantum Soft Covering} \label{sec:covering}
In this section, we consider a classical-quantum state $\rho_{XB} = \sum_{x\in\mathcal{X}} p_X(x) |x\rangle \langle x|\otimes \rho_B^x$ be a classical-quantum state. We assume that $\rho_B$ is invertible and the Hilbert space $\mathcal{H}_B$ is finite dimensional. Let $\mathcal{C}$ be a random codebook where each codeword $x\in \mathcal{X}$ is drawn independently according to distribution $p_X$.
The goal of quantum soft covering is to approximate the state $\rho_B$ using the (random) codebook-induced state $\frac{1}{|\mathcal{C}|}\sum_{x\in\mathcal{C}} \rho_B^x$.
We define the \emph{minimal random codebook size} for an $\varepsilon$-covering as
\begin{align}
    M^\varepsilon(X{\,:\,}B)_\rho := \inf\left\{  M \in \mathds{N}: |\mathcal{C}|\leq M \wedge  
    \frac12 \mathds{E}_{\mathcal{C}\sim p_X^{\otimes M} } \left\| \frac{1}{|\mathcal{C}|}\sum_{x\in\mathcal{C}} \rho_B^x - \rho_B \right\|_1 \leq \varepsilon \right\}.
\end{align}
The main result of this section is to prove the following one-shot characterization of the operational quantity $M^\varepsilon(X{\,:\,}B)_\rho$.

\begin{theo}[One-shot characterization for quantum soft covering]\label{theo:one-shot_second-order_covering}
	Given a classical-quantum state $\rho_{XB}$
	, for every $0<\varepsilon<1$ and $0< c < \delta< \frac{\varepsilon}{3}\wedge\frac{(1-\varepsilon)}{2}$, we have 
	\begin{align} \label{eq:one-shot_covering0}
	I_\textnormal{h}^{1-\varepsilon - 2\delta}\left(X{\,:\,}B\right)_{\rho} - \log \frac{1+c}{c\delta} - \log \frac{\varepsilon + c}{\delta-c} \leq
    \log M^\varepsilon(X{\,:\,}B)_\rho 
    &\leq I_\textnormal{h}^{1-\varepsilon + 3\delta} \left(X{\,:\,}B\right)_{\rho}  + \log \frac{\nu^2}{\delta^4}.
\end{align}
    Here, $I_\textnormal{h}^{\varepsilon}$ is the  $\varepsilon$-hypothesis testing information defined in \eqref{eq:one-shot_entropy} and
    $\nu = |\textnormal{\texttt{spec}}(\rho_B)|$.
\end{theo}
\begin{remark}
    In Theorem~\ref{theo:one-shot_second-order_covering}, we express the operational quantity $\log M^\varepsilon(X{\,:\,}B)_\rho$ in terms of the $(1-\varepsilon)$-hypothesis testing information.
    However, the lower bound can be improved to $D_\textnormal{s}^{1-\varepsilon - \delta}(\rho_{XB} \,\|\,\rho_X\otimes \rho_B)$ and the upper bound can be improved to ${D_\textnormal{s}^{1-\varepsilon + \delta}\left(\mathscr{P}_{\rho_B}\left[\rho_{XB}\right] \,\|\,\rho_X\otimes \rho_B\right)}$
    (both with additional additive logarithmic terms).
\end{remark}

In the scenario where the underlying state is identical and independently prepared, i.e.~$\rho_{XB}^{\otimes n}$, the established one-shot characterization, Theorem~\ref{theo:one-shot_second-order_covering}, gives the following second-order asymptotics of the logarithmic random codebook size, $\log M^\varepsilon(X^n{\,:\,}B^n)_{\rho^{\otimes n}}$, as a function of blocklength $n$, in which the optimal second-order rate is obtained.

\begin{prop}[Second-order rate for quantum soft covering] \label{prop:second-order_covering}
For every classical-quantum state $\rho_{XB}$ and $0<\varepsilon<1$, we have
\begin{align}
    \log M^\varepsilon(X^n{\,:\,}B^n)_{\rho^{\otimes n}} = n I(X{\,:\,}B)_\rho - \sqrt{n V(X{\,:\,}B) } \Phi^{-1}(\varepsilon) + O\left(\log n\right).
\end{align}
\end{prop}
\begin{proof}
Since $|\textnormal{\texttt{spec}}(\rho_B^{\otimes n})|\leq (n+1)^{|\mathcal{H}_B|-1}$ for $|\mathcal{H}_B|$ being the rank of $\rho_B$, the additive terms grow of order $O(\log n)$.
Then applying Theorem~\ref{theo:one-shot_second-order_covering}, with $\delta_n = n^{-\sfrac{1}{2}}$ and $c_n = \frac12 n^{-\sfrac{1}{2}}$, together with the second-order expansion of the hypothesis testing information, Lemma~\ref{lemm:second}, proves our claim.
\end{proof}

{ 
Moreover, the one-shot characterization, can be extended to the moderate deviation regime \cite{CTT2017, CH17}; namely, we derive the optimal rates of the minimal required random codebook size when the error approaches $0$ or $1$ moderately.

\begin{prop}[Moderate deviations for quantum soft covering] \label{prop:moderate_covering}
For every classical-quantum state $\rho_{XB}$ and every moderate sequence $(a_n)_{n\in\mathds{N}}$ satisfying \eqref{eq:an} and $\varepsilon_n :=  \mathrm{e}^{-na_n^2}$, we have 
\begin{align} 
\begin{dcases}
    \frac1n \log M^{\varepsilon_n} (X^n{\,:\,}B^n)_{\rho^{\otimes n}} = I(X{\,:\,}B)_\rho + \sqrt{2 V(X{\,:\,}B) } a_n + o\left(a_n\right); \\
    \frac1n \log M^{1-\varepsilon_n} (X^n{\,:\,}B^n)_{\rho^{\otimes n}} = I(X{\,:\,}B)_\rho - \sqrt{2 V(X{\,:\,}B) } a_n + o\left(a_n\right). \\
\end{dcases}
\end{align}
\end{prop}
\begin{proof}
    We prove the first assertion.
    For every moderate deviation sequence $(a_n)_{n\in\mathds{N}}$, we let $\delta_n = \frac14\mathrm{e}^{-na_n^2}$ satisfying \eqref{eq:an} and let $c_n = \frac12 \delta_n$.
    Then, $\varepsilon_n- 2 \delta_n$ or $\varepsilon + 3 \delta_n$ can be viewed as  $\mathrm{e}^{-n b_n^2}$ for another moderate deviation sequence $(b_n)_{n\in\mathds{N}}$; we have $b_n = a_n + (b_n-a_n) = o(a_n)$ (see e.g.~\cite[\S 7.2]{CHD+21}),
    $\frac{1}{n} \log \delta_n = -a_n^2 = o(a_n)$ and $\frac{1}{n} \log |\textnormal{\texttt{spec}}(\rho_B^{\otimes n})| =  O\left(\frac{\log n}{n}\right) = o(a_n)$. 
Then, applying Theorem~\ref{theo:one-shot_second-order_covering} together with Lemma~\ref{lemm:moderate} leads to our first claim of the moderate derivation for quantum soft covering.
The second line follows similarly.
\end{proof}
}

\begin{remark}
We remark that the upper bound may be viewed as a quantum generalization of a classical result by Hayashi \cite[Lemma 2]{Hay06_resolvability}.
However, such a generalization is non-trivial due to difficulties of non-commutativity.
\end{remark}

The proofs of the one-shot achievability (i.e.~upper bound) and converse (i.e.~lower bound) of Theorem~\ref{theo:one-shot_second-order_covering} are presented in Section~\ref{sec:achievability_covering} and \ref{sec:converse_covering}, respectively.

\subsection{Direct Bound} \label{sec:achievability_covering}
	We prove the upper bound on $\log M^\varepsilon(X:B)_\rho $ here.
 \begin{proof}[Proof of achievability of {Theorem}~\ref{theo:one-shot_second-order_covering}]	
	Throughout the proof, we use the short notation: $\rho_x \equiv \rho_B^x$ and $M \equiv |\mathcal{C}|$. 
	For every $x\in\mathcal{C}$, we define a projection $\Pi_x := \left\{ \mathscr{P}_{\rho_B}\left[\rho_x\right] \leq c \rho_B \right\}$ and its complement $\Pi_x^{\mathrm{c}}:= \mathds{1}_B - \Pi_x $.
	
	We claim that for any random codebook $\mathcal{C}$ with its codeword independently drawn according to distribution $p_X$ and for any $c>0$, we have
	\begin{align}\label{eq:one-shot_covering}
		\frac12 \mathds{E}_{\mathcal{C}} \left\| \frac{1}{|\mathcal{C}|}\sum_{x\in\mathcal{C}} \rho_B^x - \rho_B \right\|_1 \leq \Tr\left[\rho_{XB}\left\{ \mathscr{P}_{\rho_B}[\rho_{XB}] > c \rho_X\otimes \rho_B \right\} \right] + \sqrt{ \frac{|\textnormal{\texttt{spec}}(\rho_B)|c}{|\mathcal{C}|} }.
	\end{align}
	Then, let $\delta \in (0,\varepsilon)$ and choose
\begin{align}
c = \exp\left\{D_\text{s}^{1-\varepsilon + \delta} \left( \mathscr{P}_{\rho_B}[\rho_{XB}]  \,\|\, \rho_X\otimes \rho_B \right) + \xi \right\}
\end{align}
for some small $\xi>0$.
By definition of the $\varepsilon$-information spectrum divergence \eqref{eq:Ds},
\begin{align} \label{eq:upper_1}
\Tr\left[\rho_{XB}\left\{ \mathscr{P}_{\rho_B}[\rho_{XB}] > c \rho_X\otimes \rho_B \right\} \right] =
\Tr\left[\mathscr{P}_{\rho_B}[\rho_{XB}]\left\{ \mathscr{P}_{\rho_B}[\rho_{XB}] > c \rho_X\otimes \rho_B \right\} \right]
< \varepsilon - \delta.
\end{align}
Letting
\begin{align}
|\mathcal{C}| = \ceil[\Big]{ {|\texttt{spec}(\rho_B)|c}{\delta^{-2}} },
\end{align}
we obtain the $\varepsilon$-covering, i.e.~
\begin{align}
    \frac12 \mathds{E}_{\mathcal{C} } \left\| \frac{1}{|\mathcal{C}|}\sum_{x\in\mathcal{C}} \rho_B^x - \rho_B \right\|_1 \leq \varepsilon,
\end{align}
and then together \eqref{eq:one-shot_covering} we have the following upper bound on $\log M^\varepsilon(X:B)_\rho $:
\begin{align}
    \log M^\varepsilon(X{\,:\,}B)_\rho &\leq D_\text{s}^{1-\varepsilon + \delta} \left( \mathscr{P}_{\rho_B}[\rho_{XB}]  \,\|\, \rho_X\otimes \rho_B \right) + \xi + \log |\texttt{spec}(\rho_B)| - 2 \log \delta \\
    &\leq I_\text{h}^{1-\varepsilon + 3\delta} \left(X{\,:\,}B\right)_{\rho} + \xi + 2 \log |\texttt{spec}(\rho_B)| - 4 \log \delta,
\end{align}
where we have used Lemma~\ref{lemm:relation} in the last inequality.
Since $\xi>0$ is arbitrary, we take $\xi\to 0$ to obtain the upper bound in \eqref{eq:one-shot_covering0}.
	
	Now, we move on to prove \eqref{eq:one-shot_covering}.
	We first prove the case where all $\{\rho_x\}_{x}$ are invertible.
	Use triangle inequality of the norm $\|\cdot \|_1$, we obtain
	\begin{align}
		\frac12 \mathds{E}_{\mathcal{C}} \left\| \frac{1}{|\mathcal{C}|}\sum_{x\in\mathcal{C}} \rho_x - \rho_B \right\|_1 &=
		\frac12 \mathds{E}_{\mathcal{C}} \left\| \frac{1}{|\mathcal{C}|}\sum_{x\in\mathcal{C}} \rho_x(\Pi_x^\mathrm{c} + \Pi_x) - \mathds{E}_\mathcal{C}\left[ \frac{1}{|\mathcal{C}|}\sum_{x\in\mathcal{C}} \rho_x(\Pi_x^\mathrm{c} + \Pi_x)\right] \right\|_1 \\
		&\leq 
		\frac12 \mathds{E}_{\mathcal{C}} \left\| \frac{1}{|\mathcal{C}|}\sum_{x\in\mathcal{C}} \rho_x \Pi_x^\mathrm{c} - \mathds{E}_\mathcal{C}\left[ \frac{1}{|\mathcal{C}|}\sum_{x\in\mathcal{C}} \rho_x \Pi_x^\mathrm{c} \right] \right\|_1 \notag \\
		&\quad + \frac12 \mathds{E}_{\mathcal{C}} \left\| \frac{1}{|\mathcal{C}|}\sum_{x\in\mathcal{C}} \rho_x \Pi_x - \mathds{E}_\mathcal{C}\left[ \frac{1}{|\mathcal{C}|}\sum_{x\in\mathcal{C}} \rho_x \Pi_x\right] \right\|_1 \\
		\begin{split} \label{eq:one-shot_second-order1}
		&\leq \frac12 \mathds{E}_{\mathcal{C}} \left\| \frac{1}{|\mathcal{C}|}\sum_{x\in\mathcal{C}} \Pi_x^\mathrm{c}\rho_x \Pi_x^\mathrm{c} - \mathds{E}_\mathcal{C}\left[ \frac{1}{|\mathcal{C}|}\sum_{x\in\mathcal{C}} \Pi_x^\mathrm{c}\rho_x \Pi_x^\mathrm{c} \right] \right\|_1 \\
		&\quad + \frac12 \mathds{E}_{\mathcal{C}} \left\| \frac{1}{|\mathcal{C}|}\sum_{x\in\mathcal{C}} \Pi_x \rho_x \Pi_x^\mathrm{c} - \mathds{E}_\mathcal{C}\left[ \frac{1}{|\mathcal{C}|}\sum_{x\in\mathcal{C}} \Pi_x \rho_x \Pi_x^\mathrm{c} \right] \right\|_1  \\
		&\quad + \frac12 \mathds{E}_{\mathcal{C}} \left\| \frac{1}{|\mathcal{C}|}\sum_{x\in\mathcal{C}} \rho_x \Pi_x - \mathds{E}_\mathcal{C}\left[ \frac{1}{|\mathcal{C}|}\sum_{x\in\mathcal{C}} \rho_x \Pi_x\right] \right\|_1.
		\end{split}
	\end{align}
	The first term in \eqref{eq:one-shot_second-order1} can be further bounded using triangle inequality again as:
	\begin{align}
		\frac12 \mathds{E}_{\mathcal{C}} \left\| \frac{1}{|\mathcal{C}|}\sum_{x\in\mathcal{C}} \Pi_x^\mathrm{c}\rho_x \Pi_x^\mathrm{c} - \mathds{E}_\mathcal{C}\left[ \frac{1}{|\mathcal{C}|}\sum_{x\in\mathcal{C}} \Pi_x^\mathrm{c}\rho_x \Pi_x^\mathrm{c} \right] \right\|_1
		&\leq \frac12 \mathds{E}_{\mathcal{C}} \left\| \frac{1}{|\mathcal{C}|}\sum_{x\in\mathcal{C}} \Pi_x^\mathrm{c}\rho_x \Pi_x^\mathrm{c} \right\|_1 + \frac12 \left\|\mathds{E}_\mathcal{C}\left[ \frac{1}{|\mathcal{C}|}\sum_{x\in\mathcal{C}} \Pi_x^\mathrm{c}\rho_x \Pi_x^\mathrm{c} \right] \right\|_1 \\
		&\leq \mathds{E}_{\mathcal{C}} \left\| \frac{1}{|\mathcal{C}|}\sum_{x\in\mathcal{C}} \Pi_x^\mathrm{c}\rho_x \Pi_x^\mathrm{c} \right\|_1 \\
		&= \mathds{E}_{x\sim p_X} \Tr\left[ \rho_x \left\{\mathscr{P}_{\rho_B}\left[\rho_x\right] > c\rho_B \right\} \right] \\
		&= \Tr\left[\rho_{XB}\left\{\mathscr{P}_{\rho_B}\left[\rho_{XB}\right]> c \rho_X\otimes \rho_B \right\} \right]. \label{eq:one-shot_second-order2}
	\end{align}

	Next, we bound the second term in \eqref{eq:one-shot_second-order1}:
	\begin{align}
		&\frac12 \mathds{E}_{\mathcal{C}} \left\| \frac{1}{|\mathcal{C}|}\sum_{x\in\mathcal{C}} \Pi_x \rho_x \Pi_x^\mathrm{c} - \mathds{E}_\mathcal{C}\left[ \frac{1}{|\mathcal{C}|}\sum_{x\in\mathcal{C}} \rho_x \Pi_x^\mathrm{c} \right] \right\|_1 
		= \frac12 \mathds{E}_{\mathcal{C}} \left\| \frac{1}{|\mathcal{C}|}\sum_{x\in\mathcal{C}} \Pi_x^\mathrm{c} \rho_x \Pi_x - \mathds{E}_\mathcal{C}\left[ \frac{1}{|\mathcal{C}|}\sum_{x\in\mathcal{C}} \Pi_x^\mathrm{c} \rho_x \Pi_x \right] \right\|_1 \\
		&= \frac12 \mathds{E}_{\mathcal{C}} \Tr \sqrt{ \left(\frac{1}{|\mathcal{C}|}\sum_{x\in\mathcal{C}} \Pi_x^\mathrm{c} \rho_x \Pi_x - \mathds{E}_\mathcal{C}\left[ \frac{1}{|\mathcal{C}|}\sum_{x\in\mathcal{C}} \Pi_x^\mathrm{c}\rho_x\Pi_x  \right]\right)^\dagger \left( \frac{1}{|\mathcal{C}|}\sum_{x\in\mathcal{C}} \Pi_x^\mathrm{c} \rho_x \Pi_x - \mathds{E}_\mathcal{C}\left[ \frac{1}{|\mathcal{C}|}\sum_{x\in\mathcal{C}} \Pi_x^\mathrm{c} \rho_x \Pi_x \right] \right) } \\
		&\overset{(a)}{\leq} \frac12  \Tr \sqrt{ \mathds{E}_{\mathcal{C}} \left(\frac{1}{|\mathcal{C}|}\sum_{x\in\mathcal{C}} \Pi_x^\mathrm{c} \rho_x \Pi_x - \mathds{E}_\mathcal{C}\left[ \frac{1}{|\mathcal{C}|}\sum_{x\in\mathcal{C}} \Pi_x^\mathrm{c}\rho_x\Pi_x  \right]\right)^\dagger \left( \frac{1}{|\mathcal{C}|}\sum_{x\in\mathcal{C}} \Pi_x^\mathrm{c} \rho_x \Pi_x - \mathds{E}_\mathcal{C}\left[ \frac{1}{|\mathcal{C}|}\sum_{x\in\mathcal{C}} \Pi_x^\mathrm{c} \rho_x \Pi_x \right] \right) } \\
		&\overset{(b)}{=} \frac12  \Tr \sqrt{ \mathrm{Var}_{\mathcal{C}} \left[ \frac{1}{|\mathcal{C}|}\sum_{x\in\mathcal{C}} \Pi_x^\mathrm{c} \rho_x \Pi_x \right] } \\
		&\overset{(c)}{=} \frac12  \Tr \sqrt{ \frac{1}{|\mathcal{C}|} \mathrm{Var}_{x\sim p_X} \left[ \Pi_x^\mathrm{c} \rho_x \Pi_x \right] } \\
		&\overset{(d)}{\leq} \frac12  \Tr \sqrt{ \frac{1}{|\mathcal{C}|} \mathds{E}_{x\sim p_X} \left[ \Pi_x \rho_x \Pi_x^\mathrm{c} \rho_x \Pi_x \right] } \\
		&\overset{(e)}{\leq} \frac12  \Tr \sqrt{ \frac{1}{|\mathcal{C}|} \mathds{E}_{x\sim p_X} \left[ \Pi_x \rho_x^2 \Pi_x \right] }.
	\end{align}
	Here, in (a) we applied Jensen's inequality, Lemma~\ref{lemm:Jensen}, with expectation $\mathds{E}_{\mathcal{C}}$ and the operator concavity of square-root;
	in (b) we denoted a {matrix-valued variance} for a random matrix $H$ as $\mathrm{Var}[H]:= \mathds{E}[H^\dagger H] - (\mathds{E}[H])^\dagger \mathds{E}[H]$ (see e.g.~\cite[\S 2]{Tro15});
	in (c) we applied the mutual independence of the codewords in the random codebook;
	in (d) we used the operator monotone of square-root and $\mathrm{Var}[H] \leq \mathds{E}[H^\dagger H]$;
	and (e) follows from $\Pi_x^\mathrm{c} \leq \mathds{1}_B$ and the operator monotonicity of square-root. 
	
	Applying the same reasoning on the third term of \eqref{eq:one-shot_second-order1}, we thus upper bound the sum of the second and the third term of \eqref{eq:one-shot_second-order1} by
	\begin{align} \label{eq:one-shot_second-order3}
		\Tr \sqrt{ \frac{1}{|\mathcal{C}|} \mathds{E}_{x\sim p_X} \left[ \Pi_x \rho_x^2 \Pi_x \right] }.
	\end{align}

	To further upper bound this term, we use Jensen's inequality, Lemma~\ref{lemm:Jensen}, with the pinching map $\mathscr{P}_{\rho_B}$ and the operator concavity of square-root to have
	\begin{align}
		\Tr \sqrt{ \mathds{E}_{x\sim p_X} \left[ \Pi_x \rho_x^2 \Pi_x \right] } &= \Tr \mathscr{P}_{\rho_B}\sqrt{ \mathds{E}_{x\sim p_X} \left[ \Pi_x \rho_x^2 \Pi_x \right] } \\
		&\leq \Tr\sqrt{ \mathscr{P}_{\rho_B} \left[ \mathds{E}_{x\sim p_X} \left[ \Pi_x \rho_x^2 \Pi_x \right]\right] } \\
		&= \Tr \left[ \rho_B \sqrt{ \mathscr{P}_{\rho_B} \left[ \mathds{E}_{x\sim p_X} \left[ \Pi_x \rho_x^2 \Pi_x \right]\right] \rho_B^{-2} } \right] \\
		&\overset{(a)}{\leq} \sqrt{ \Tr\left[ \mathds{E}_{x\sim p_X} \left[ \Pi_x \rho_x^2 \Pi_x \right]\right] \rho_B^{-1} } \\
		&= \sqrt{ \mathds{E}_{x\sim p_X} \Tr\left[\rho_x^2 \Pi_x \rho_B^{-1} \Pi_x \right]  }, \label{eq:one-shot_second-order4}
	\end{align}
	where (a) follows from Jensen's inequality, Lemma~\ref{lemm:Jensen}, with the functional $\Tr[\rho_B (\, \cdot \,)]$ and the operator concavity of square-root.

	Now, since 
	\begin{align}
		\Pi_x = \left\{ \mathscr{P}_{\rho_B}[\rho_x] \leq c \rho_B \right\} = \left\{ \rho_B^{-1} \leq c \left(\mathscr{P}_{\rho_B}[\rho_x]\right)^{-1} \right\},
	\end{align}
 	we obtain
 	\begin{align}
 		\Pi_x \rho_B^{-1} \Pi_x &\leq c \Pi_x \left(\mathscr{P}_{\rho_B}[\rho_x]\right)^{-1} \Pi_x \\
 		&= c\left(\mathscr{P}_{\rho_B}[\rho_x]\right)^{-1/2}\Pi_x \left(\mathscr{P}_{\rho_B}[\rho_x]\right)^{-1/2},\\
 		\end{align}
 		where we used the fact that $\Pi_x$ commutes with $\mathscr{P}_{\rho_B}[\rho_x]$. Then for each $x$,
 			\begin{align}
 			\Tr\left[\rho_x^2 \Pi_x \rho_B^{-1} \Pi_x \right] \le &c\Tr\left[\rho_x^2 \mathscr{P}_{\rho_B}(\rho_x)^{-1/2} \Pi_x \mathscr{P}_{\rho_B}(\rho_x)^{-1/2} \right] \\ \le& c\Tr\left[\rho_x^2 \mathscr{P}_{\rho_B}(\rho_x)^{-1} \right]
 			\\ \overset{(a)}{\le}& c|\textnormal{\texttt{spec}}(\rho_B)|\Tr\left[\rho_x^2 \rho_x^{-1} \right] \\
 			=&c|\textnormal{\texttt{spec}}(\rho_B)|, \label{eq:one-shot_second-order5}
 			\end{align}
 	where in (a) we used the the pinching inequality (Lemma~\ref{lemm:pinching}), i.e.~
 \[\rho\le |\textnormal{\texttt{spec}}(\rho_B)| \mathscr{P}_{\rho_B}[\rho_x],\] and the operator monotonicity of inversion.
	Combining \eqref{eq:one-shot_second-order1}, \eqref{eq:one-shot_second-order3}, \eqref{eq:one-shot_second-order4}, and \eqref{eq:one-shot_second-order5} arrives at the desired \eqref{eq:one-shot_covering}. 
	
For the general non-invertible state $\rho_x$, we define the approximation
	\[\rho_x^\epsilon := (1-\epsilon)\rho_x+\eps\frac{\mathds{1}_{B}}{|\mathcal{H}_B|},\]
and then
	\[\rho_B^\epsilon := (1-\epsilon)\rho_B+\epsilon\frac{\mathds{1}_{B}}{|\mathcal{H}_B|}\ ,\  \rho_{XB}^\epsilon=(1-\epsilon)\rho^\epsilon_{XB}+\epsilon \rho_X \ten \frac{\mathds{1}_{B}}{|\mathcal{H}_B|}. \]
Moreover, since the pinching map $\mathscr{P}_{\rho_B}$ is unital completely positive and trace-preserving,
		\[\mathscr{P}_{\rho_B}(\rho_{XB}^\epsilon)=(1-\epsilon)\mathscr{P}_{\rho_B}(\rho_{XB})+\epsilon \rho_X\ten \frac{\mathds{1}_{B}}{|\mathcal{H}_B|} \pl.\]
It is clear that
\[ \lim_{\epsilon \to 0}\frac12 \mathds{E}_{\mathcal{C}} \left\| \frac{1}{|\mathcal{C}|}\sum_{x\in\mathcal{C}} \rho_x^\epsilon - \rho_B^\epsilon \right\|_1 =\frac12 \mathds{E}_{\mathcal{C}} \left\| \frac{1}{|\mathcal{C}|}\sum_{x\in\mathcal{C}} \rho_x - \rho_B \right\|_1.\]
For the right-hand side of the desired inequality \eqref{eq:one-shot_covering}, we have 
\[\left\| \mathscr{P}_{\rho_B}(\rho_{XB}^\epsilon)-\mathscr{P}_{\rho_B}(\rho_{XB}) \right\|_1 \le 2\epsilon\]
and for small enough $\epsilon$, the projection 
\[ \{ \mathscr{P}_{\rho_B}[\rho_{XB}^\epsilon] > c \rho_X\otimes \rho_B^\epsilon \}= \{ \mathscr{P}_{\rho_B}[\rho_{XB}] > c \rho_X\otimes \rho_B \}.\]
Indeed, because $\mathscr{P}_{\rho_B}[\rho_{XB}]$ and $\rho_X\otimes \rho_B$ are commutative, they can be viewed as functions on the finite sets of spectrum. Therefore, we have
\[\lim_{\epsilon\to 0}\Tr\left[\rho_{XB}^\epsilon\left\{ \mathscr{P}_{\rho_B}[\rho_{XB}^\epsilon] > c \rho_X\otimes \rho_B^\epsilon \right\} \right]=\Tr\left[\rho_{XB}\left\{ \mathscr{P}_{\rho_B}[\rho_{XB}] > c \rho_X\otimes \rho_B \right\} \right]\ ,\]
which proves \eqref{eq:one-shot_covering} by approximation. 
 \end{proof}

\subsection{Converse Bound} \label{sec:converse_covering}
	We prove the lower bound on on $\log M^\varepsilon(X{\,:\,}B)_\rho $ here.
\begin{proof}[Proof of converse of {Theorem}~\ref{theo:one-shot_second-order_covering}]	
Throughout this proof, we write $\rho_x \equiv \rho_B^x$ and  $M := |\mathcal{C}|$. 
For every $0<c<\delta < \frac{1-\varepsilon}{2}$ and for any realization of the random codebook $\mathcal{C}$, we choose the noncommutative quotient
\begin{align}
    \Pi = \frac{ \frac{1}{M} \sum_{\bar{x}\in\mathcal{C}} \rho_{\bar x}  }{ \frac{1}{M} \sum_{\bar x\in\mathcal{C}} \rho_{\bar x} + c^{-1} \rho_B }.
\end{align}
Lemma~\ref{lemm:1norm} then implies that
\begin{align}
    \frac12 \left\| \frac{1}{M}\sum_{x\in\mathcal{C}} \rho_B^x - \rho_B \right\|_1 \geq \Tr\left[ \frac{1}{M}\sum_{x\in\mathcal{C}} \rho_B^x \Pi \right] - \Tr\left[ \rho_B \Pi \right] \label{eq:converse_covering1}
\end{align}
Using Lemma~\ref{lemm:quiotient}-\ref{item:d}, the second term in \eqref{eq:converse_covering1} can be lower bounded as
\begin{align}
    - \Tr\left[ \rho_B \Pi \right] &= - c \Tr\left[ \frac{1}{M} \sum_{\bar x\in\mathcal{C}} \rho_{\bar x} \cdot \frac{ c^{-1}\rho_B }{\frac{1}{M} \sum_{\bar x\in\mathcal{C}} \rho_{\bar x} + c^{-1} \rho_B } \right] \\
    &\geq -c \Tr\left[ \frac{1}{M} \sum_{\bar x\in\mathcal{C}} \rho_{\bar x} \right] \\
    &= - c \,, \label{eq:converse_covering2}
\end{align}
since $\frac{ c^{-1}\rho_B }{\frac{1}{M} \sum_{\bar x\in\mathcal{C}} \rho_{\bar x} + c^{-1} \rho_B }\leq \mathds{1}_B$ by Lemma~\ref{lemm:quiotient}-\ref{item:c}.

For each $x\in\mathcal{C}$, we use (by recalling Lemma~\ref{lemm:quiotient}-\ref{item:a} \& \ref{item:b})
\begin{align}
    \Pi \geq \frac{ \rho_{x} }{ \sum_{\bar x\in\mathcal{C}} \rho_{\bar x} + c^{-1} M \rho_B }
\end{align}
to lower bound the first term in \eqref{eq:converse_covering1} as
\begin{align}
    \Tr\left[ \frac{1}{M}\sum_{x\in\mathcal{C}} \rho_B^x \Pi \right] &\geq \frac{1}{M}\sum_{x\in\mathcal{C}} \Tr\left[ \rho_x \frac{ \rho_{x} }{ \sum_{\bar x\in\mathcal{C}} \rho_{\bar x} + c^{-1} M \rho_B }\right] \\
    &= \frac{1}{M} \exp D_2^*\left( \rho_{XB}^{\mathcal{C}}  \,\|\, \rho_{X}^{\mathcal{C}}\otimes \rho_{B}^{\mathcal{C}} + c^{-1} \rho_{X}^{\mathcal{C}}\otimes \rho_B \right), \label{eq:converse_covering3}
\end{align}
where we recall the definition of the sandwiched R\'enyi divergence, \eqref{eq:sand}, and
\begin{align}
    \rho_{XB}^{\mathcal{C}} &:= \sum_{x\in\mathcal{C}} \frac{1}{M} |x\rangle \langle x| \otimes \rho_x.
\end{align}
With \eqref{eq:converse_covering1}, \eqref{eq:converse_covering2}, and \eqref{eq:converse_covering3}, we take expectation over the random codebook and use the joint convexity of $\exp D_2^*(\,\cdot\|\,\cdot)$, Lemma~\ref{lemm:joint_convexity}, to obtain
\begin{align}
    \frac12 \mathds{E}_{\mathcal{C}} \left\| \frac{1}{M}\sum_{x\in\mathcal{C}} \rho_B^x - \rho_B \right\|_1 \geq \frac{1}{M} \exp D_2^*\left(\left. \mathds{E}_{\mathcal{C}}\left[\rho_{XB}^{\mathcal{C}}\right] \right\| \mathds{E}_{\mathcal{C}}\left[ \rho_{X}^{\mathcal{C}}\otimes \rho_{B}^{\mathcal{C}} + c^{-1} \rho_{X}^{\mathcal{C}}\otimes \rho_B \right] \right) - c. \label{eq:converse_covering4}
\end{align}

Now, we apply the mutual independence between the codeword to have
\begin{align}
    \mathds{E}_{\mathcal{C}}\left[\rho_{XB}^{\mathcal{C}}\right] &= \rho_{XB}; \\
    \mathds{E}_{\mathcal{C}}\left[ \rho_{X}^{\mathcal{C}}\otimes \rho_{B}^{\mathcal{C}} \right] &= \mathds{E}_{\mathcal{C}}\left[ \frac{1}{M^2} \sum_{m, \bar{m} \in [M]} |x_m\rangle \langle x_m  |\otimes \rho_{{x}_{\bar m} }   \right] \\
    &= \mathds{E}_{\mathcal{C}}\left[ \frac{1}{M^2} \sum_{m \in [M]} |x_m\rangle \langle x_m | \otimes \rho_{{x}_m }   \right] + \mathds{E}_{\mathcal{C}}\left[ \frac{1}{M^2} \sum_{m\neq \bar{m}} |x_m\rangle \langle x_m | \otimes \rho_{{x}_{\bar m} } \right] \\
    &= \frac{1}{M} \rho_{XB} + \left( 1 - \frac{1}{M} \right) \rho_X\otimes \rho_B; \\
    c^{-1} \mathds{E}_{\mathcal{C}}\left[ \rho_{X}^{\mathcal{C}}\otimes \rho_B \right] &= c^{-1} \rho_X \otimes \rho_B.
\end{align}
By putting them together, we obtain
\begin{align}
    \varepsilon &\geq \frac{1}{M} \exp D_2^*\left( \rho_{XB}  \,\|\, M^{-1} \rho_{XB} + \left( 1 + c^{-1} - M^{-1} \right) \rho_{X}\otimes \rho_B \right) - c \\
    &\overset{(a)}{\geq} \frac{\varepsilon+\delta }{M} \left[  M^{-1} + \left( 1 + c^{-1} - M^{-1} \right) \e^{- D_\text{s}^{1-\varepsilon - \delta}(\rho_{XB} \,\|\,\rho_X\otimes \rho_B) } \right]^{-1} - c,
\end{align}
where we apply Lemma~\ref{lemm:D2toDs} in (a) with $\eta = 1 - \varepsilon - \delta$, $\lambda_1 = M^{-1}$, and $\lambda_2 = 1 + c^{-1} - M^{-1}$.
In other words, we get an lower bound on $\log M$ as
\begin{align}
    \log M &\geq D_\text{s}^{1-\varepsilon - \delta}(\rho_{XB} \,\|\,\rho_X\otimes \rho_B) - \log \left(1 + c^{-1} + M^{-1} \right) + \log \frac{\delta-c}{\varepsilon + c}\\
    &\geq D_\text{s}^{1-\varepsilon - \delta}(\rho_{XB} \,\|\,\rho_X\otimes \rho_B) - \log \left(1 + c^{-1}\right) + \log \frac{\delta-c}{\varepsilon + c}\\
    &\geq D_\text{h}^{1-\varepsilon - 2\delta}(\rho_{XB} \,\|\,\rho_X\otimes \rho_B) - \log \left(1 + c^{-1}\right) + \log \frac{\delta-c}{\varepsilon + c} + \log \delta\\
    &\geq I_\text{h}^{1-\varepsilon - 2\delta}(X{\,:\,}B)_\rho - \log \frac{1+c}{c\delta} - \log \frac{\varepsilon + c}{\delta-c},
\end{align}\vspace{-0.5em}
completing the proof.
\end{proof}

\vspace{-2em}
\section{Conclusions} \label{sec:conclusions}
The \emph{large deviation analysis} \cite{DZ98, Hay07, ANS+08, WWY14, MO17, Hao-Chung, CH16, CHT19, CHDH2-2018, Cheng2021a, Cheng2021b, CGH18} of privacy amplification against quantum side information and quantum soft covering has been investigated in previous literature \cite{Ren05, Dup21, LY21a, KL21, CG22, SGC22a}, wherein one fixes the rate or the size of $|\mathcal{Z}|$ and $|\mathcal{C}|$ and studies the optimal errors in terms of the trace distance. Also, some \emph{moderate deviation analysis} \cite{CH17, CTT2017} were studied for characterizing the minimal trace distance while the rates approach the first-order limits with certain speed \cite{CG22, SGC22a}.
In this paper, we took another perspective---what are the optimal rates when the trace distances are upper bounded by a constant $\varepsilon \in (0,1)$. This corresponds to the so-called \emph{small error regime} \cite{Str62, Hay09b, PPV10} or the \emph{non-vanishing error regime} \cite{Tan14}.
We establish the second-order rates for fixed $\varepsilon\in (0,1)$ and establish the optimal rates when trace distances vanishes no faster than $O(\sfrac{1}{\sqrt{n}})$. 

In light of the duality between smooth min- and max-entropies, the purified distance has been recognized as an appropriate distance measure \cite{TCR10, TH13, Tsurumaru2021EquivalenceOT, ABJ+20, KL21}. Our work suggests that if one considers the trace distance as the performance benchmark without going into the smooth entropy framework \cite{WH13, Hay13, TH13, Hay16, HW16, ABJ+20}, the conditional hypothesis testing entropy and the hypothesis testing information\footnote{In Ref.~\cite{TH13}, it was shown that up to second-order terms, $D_\text{h}^{1-\varepsilon}$ scales as the relative entropy version of the smooth min-entropy $D_\text{max}^{\sqrt{\varepsilon}}$. Although the two quantities are asymptotically equivalent, they arise in very different proof methodologies.
} are the natural one-shot characterizations\footnote{
In Ref.~\cite{Dup21}, Dupuis asked is it possible to use the R\'enyi-type entropies/information to characterize operational quantities in one-shot information theory.
We would like to point out that although there are essentially no differences between the three deviation regimes in the one-shot setting, there are at least two different types of operational quantities of interest; their characterizations might be different.
As observed in \cite{Dup21} and our previous works \cite{SGC22a, CG22}, indeed, the R\'enyi-type quantities are more favorable in characterizing the optimal error given a fixed size or cardinality such as $|\mathcal{Z}|$ and $|\mathcal{C}|$ considered in this paper.
On the other hand, if one concerns the size or cardinality given a fixed error, the hypothesis-testing-type quantities or the information-spectrum-type quantities might be more direct for characterizations.
}
. 
An interesting open question is comparison between the conditional hypothesis testing entropy with the partially trace-distance-smoothed min-entropy \cite{ABJ+20}.

\section*{Acknowledgement}
H.-C.~Cheng would like to thank Kai-Min Chung for his insightful discussions, and also thank Marco Tomamichel for helpful comments and pointing out relevant references.
Y.-C.~Shen and H.-C.~Cheng are supported by the Young Scholar Fellowship (Einstein Program) of the Ministry of Science and Technology in Taiwan (R.O.C.) under Grant MOST 110-2636-E-002-009, and are supported by the Yushan Young Scholar Program of the Ministry of Education in Taiwan (R.O.C.) under Grant NTU-110V0904,  Grant NTU-CC-111L894605, and Grand NTU-111L3401.

{\larger
\bibliographystyle{myIEEEtran}
\bibliography{reference.bib}
}

\end{document}